\documentclass[12pt,a4paper,oneside,reqno]{amsart}
\usepackage[pdfdisplaydoctitle=true,colorlinks=true,urlcolor=blue,citecolor=blue,linkcolor=blue,pdfstartview=FitH,pdfpagemode=None,bookmarksnumbered=true]{hyperref}
\usepackage{amssymb}
\usepackage{amsmath}
\usepackage{amsthm}
\usepackage{color}
\usepackage{hyperxmp} 
\usepackage{cmap} 
\usepackage{graphicx}
\usepackage{framed, color}
\usepackage{xcolor}
\usepackage{subcaption}
\usepackage{lipsum} 
\usepackage{enumerate}

\usepackage[draft]{todonotes}   

\DeclareMathOperator{\spn}{span}

\newtheorem{theorem}{Theorem}[section]

\newtheorem{lemma}[theorem]{Lemma}
\newtheorem{proposition}[theorem]{Proposition}

\theoremstyle{definition}
\newtheorem{definition}[theorem]{Definition}
\newtheorem{assumption}[theorem]{Assumption}

\theoremstyle{remark}
\newtheorem{remark}[theorem]{Remark}

\newcommand{\ri}{\mathrm{i}}

\newcommand{\Eig}{\mathrm{Eig}}
\newcommand{\Ran}{\mathrm{Range}}

\begin{document}
	
	\title[]{Exact Model Reduction for Damped-Forced Nonlinear Beams:\\
		 An Infinite-Dimensional Analysis}
	\author[F. Kogelbauer and G. Haller]{Florian Kogelbauer and George Haller}
	\address{Institute for Mechanical Systems, ETH Z\"{u}rich, Leonhardstrasse 21, 8092 Z\"{u}rich, Switzerland}
	\email{floriank@ethz.ch and georgehaller@ethz.ch}
	
	
	\date{\today}%
	
	
\begin{abstract}
We use  invariant manifold results on Banach spaces to conclude the existence of spectral submanifolds (SSMs) in a class of nonlinear, externally forced beam oscillations. SSMs are the smoothest nonlinear extensions of spectral subspaces of the linearized beam equation. Reduction of the governing PDE to SSMs provides an explicit low-dimensional model which captures the correct asymptotics of the full, infinite-dimensional dynamics. Our approach is general enough to admit extensions to other types of continuum vibrations. The model-reduction procedure we employ also gives guidelines for a mathematically self-consistent modeling of damping in PDEs describing structural vibrations.

\end{abstract}

\maketitle

\section{Introduction}
Most model-reduction techniques, such as the Proper Orthogonal Decomposition (POD, cf. Holmes, Lumley and Berkooz \cite{holmes_lumley_berkooz_1996}) or the method of Modal Derivatives (cf. Idelsohn and Cardona \cite{IDELSOHN1985253} and Rutzmoser et al. \cite{PUB_2014_Rutzmoser_ISMA.pdf}), project the dynamics of a nonlinear system onto a linear subspace or a quadratic manifold, respectively. These subspaces or surfaces, however, are generally not invariant under the flow, i.e., a trajectory starting at a point on the plane used by POD will leave the plane as time evolves, cf. Figure \ref{projection}. This lack of invariance limits the reliability of model reduction to regions that are close enough to linear evolution. At larger distances from the equilibrium, nonlinear effects of the underlying model become invariably more dominant and the accuracy of model reduction is quickly lost.\\

\begin{figure}[h]
	\centering
	\includegraphics[scale=0.6]{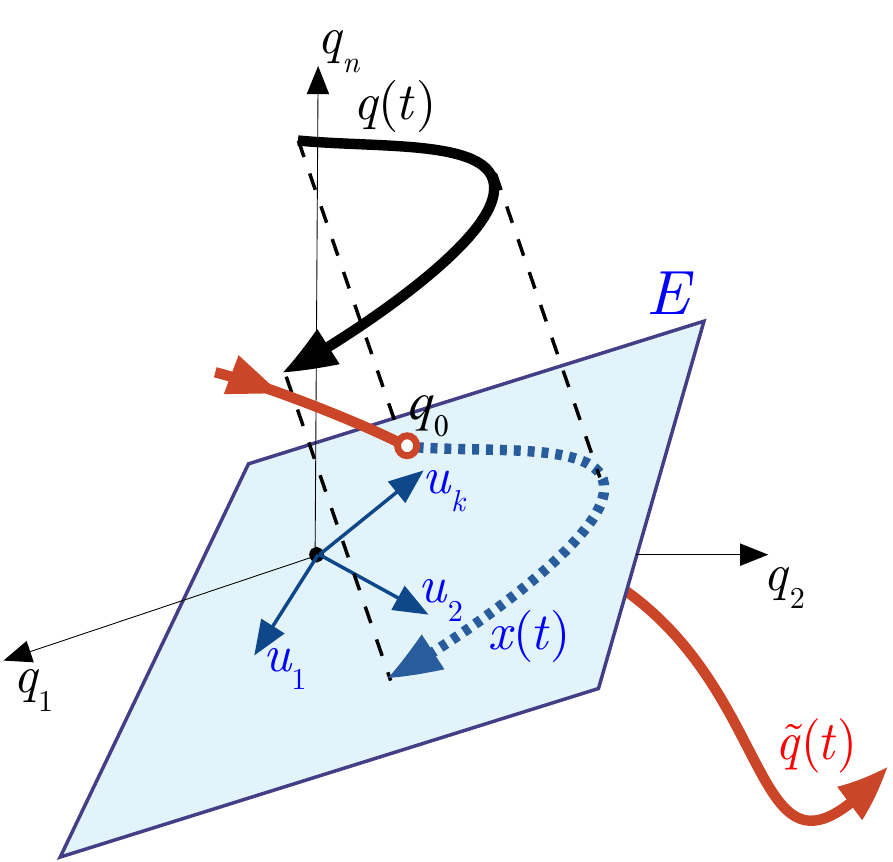}
	\caption{The classic idea of model reduction: projection of the equations of an n-dimensional dynamical system onto a lower-dimensional subspace E spanned by an appropriate basis $u_{1},\ldots u_{k}$ with $k\leq n$. The unverified hope is then that trajectories $q(t)$ of the full system project to trajectories $x(t)$ of the reduced model. The subspace $E$, however, is generally not invariant: a full trajectory $\tilde{q}(t)$ starting from the initial condition$ q_{0}\in E$ will leave $E$, and will not project to $x(t)$. (Adopted from\cite{Haller2016}).}
	\label{projection}
\end{figure}

A seminal idea to remedy the above deficiency in nonlinear vibrations is due to Shaw and Pierre  \cite{SHAW199385}. These authors proposed to reduce the full nonlinear dynamics near equilibria to invariant manifolds that are tangent to spectral subspaces of the linearized dynamical system. In lightly damped structural dynamics problems, the relevant equilibria are asymptotically stable fixed points with complex eigenvalues. In that case, the Shaw–-Pierre approach seeks invariant manifolds tangent to two-dimensional eigenspaces of the linearized dynamics. While the existence, uniqueness, smoothness and robustness of such manifolds has remained unclear in the mechanics literature, formal Taylor expansions for such manifolds have been found very effective in capturing the reduced dynamics in a number of examples (see  Kerschen et al. \cite{Kerschen2009170} and Avramov and Mikhlin  \cite{Avramov2010,Avramov2013}  for recent reviews). Shaw and Pierre \cite{SHAW1994319} also extended their original approach formally to PDEs describing continuum vibrations.\\

The invariant manifolds envisioned by Shaw and Pierre turn out to be highly non-unique and non-smooth even in linear systems, as several authors have observed recently (see, e.g., Neild et al. \cite{Neild20140404} and Cirillo et al. \cite{Cirillo2016284,Cirillo2015}). This observation carries over to the full nonlinear setting, as one can conclude from the powerful abstract results of Cabr\'{e}, Fontich  de la Llave \cite{Cab2003} on invariant manifold tangent to spectral subspaces of maps on Banach spaces.\\ 

Very recently, however, Haller and Ponsioen \cite{Haller2016} introduced the notion of a spectral submanifold (SSM) for finite-dimensional oscillation problems. An SSM is the smoothest nonlinear continuation of a spectral subspace of a linear dynamical system under the addition of nonlinear terms. Based on an analysis of the linearized spectrum, one can use the Cabr\'{e}–-Fontich–-de la Llave results to construct SSMs in appropriate smoothness classes and develop a Taylor expansion or unique internal parametrization for them. Haller and Ponsioen \cite{Haller2016} showed applications of this mathematically exact, nonlinear model reduction procedure for lower-dimensional mechanical models. Subsequently, Jain, Tiso and Haller \cite{Jain2017} carried out an SSM-based model reduction on a higher-dimensional finite-element model of a von K\'{a}rm\'{a}n beam.\\

Up to this point, however, no SSM-reduction has been carried out for continuum vibration problems. The existence of such a reduction is important to clarify for several reasons. First, all practical structural vibration problems arise from discretizations of PDEs. While SSM-based model reduction for such discretized problems has been demonstrated (cf. Jain, Tiso and Haller \cite{Jain2017}), it is not immediately clear how closely these reduced models reproduce features of the original infinite-dimensional physical structure. In fact, as we will see later, convergence of SSM-based reduced models obtained from discretized PDEs under refinement of the discretization is by no means guaranteed under the most commonly used damping models. Second, the existence of exact, SSM-based reduced models for the PDE enables one to avoid numerical experimentation with discretizations of the PDE, and proceed instead directly to an exact lower-dimensional model that is guaranteed to capture the correct asymptotics of the PDE. Third, experimentally observed simple dynamics on SSMs raises the possibility of accurate parameter identification for the full PDE.\\

In the present paper, we carry out an exact, SSM-based model reduction procedure for the first time for a nonlinear continuum vibration problem, clarifying the conditions under which the finite-dimensional invariant manifolds envisioned by Shaw and Pierre \cite{SHAW1994319} exist and smoothly persist. We believe that this is also the first infinite-dimensional application of the abstract invariant manifold results of Cabr\'{e}, Fontich  de la Llave \cite{Cab2003}. Specifically, we consider a Rayleigh beam model together with a damping proportional to the bending rate of the beam as well as viscous damping introduced by external, linear dampers. We will assume that the beam interacts with its nonlinear foundation and is also possibly subject to time-periodic external forcing. We choose the damping mechanism carefully so that the spectrum of linearized flow is contained in the unit circle, stays bounded away from zero and has monotonically decaying real parts. These properties turn out to be crucial for the existence of an SSM-based reduced-order model. For a detailed discussion of the mechanical model used for our analysis, we refer to Section \ref{mech}, while for a general discussion of damping mechanisms in beams, we refer to \cite{doi:10.1137/1.9781611970982.ch4}.\\

As an example, we explicitly compute up to third order an exact reduced model on a two-dimensional SSM of the full beam equation with cubic nonlinearity. We perform this computation both for zero  and for sinusoidal external forcing. The SSMs obtained in this fashion serve effectively as slow manifolds for the beam, even though no explicit slow-fast decomposition is available for the underlying mechanical model. In our analysis, we take advantage of the smallness of both the viscous damping and the internal damping, which permits us to avoid small denominators in the Taylor approximation of the SSM. In finite dimensions, similar calculations have been carried out by Szalai, Erhardt and Haller \cite{Szalai2016} in the absence of time-periodic forcing. We utilize infinite-dimensional analogs of their explicit formulas in computing the reduced model on the slow SSM.\\

\section{The Mechanical Model}\label{mech}
Consider a homogeneous thin beam with fixed end points at $x=0$ and $x=\pi$, respectively.
We assume that the beam is made of an elastic material and that the deflections from the unforced position of rest are comparably small. We will additionally assume that a small beam element is endowed with a small mass moment of inertia, so that the Rayleigh beam theory applies, cf. \cite{doi:10.1137/1.9781611970982.ch4}.\\
Under the assumption that the beam is composed of perfectly inelastic fibers, we may add a frequency-dependent damping to the equations of motions. This can be interpreted as a lateral force acting on the beam which is proportional to the bending rate. We refer to \cite{doi:10.1137/1.9781611970982.ch4}, page 126 and pages 135-140, especially page 139, for the details of Rayleigh beam theory and frequency-dependent damping.\\
Further, we impose that the beam is supported on a nonlinear foundation, like rubber or springs with higher-order stiffness, and we assume that the beam is subject to a weak time-periodic external forcing, cf. Figure \ref{BeamGraphic}.
\begin{figure}[h]
	\centering
	\includegraphics[scale=0.4]{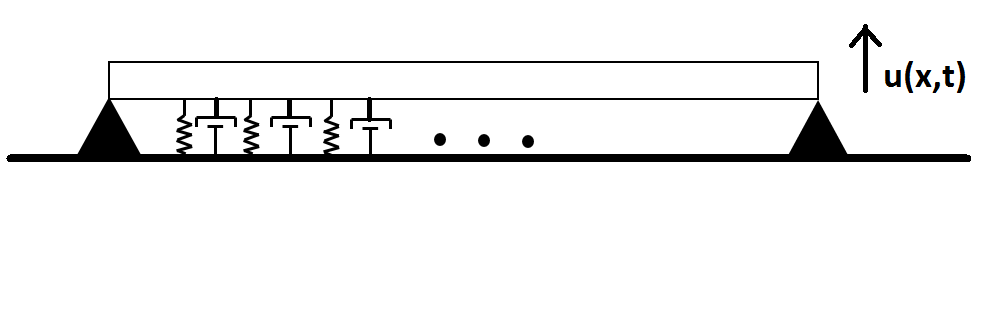}
	\caption{Simply supported Rayleigh beam on a bed of nonlinear springs and linear dampers.}
	\label{BeamGraphic}
\end{figure}

The equation of motion for the vertical displacement $u$ with initial configuration $u_0(x)$ and initial velocity $v_0(x)$ in our model becomes

\begin{equation}\label{mainequ}
\begin{cases}
&u_{tt}-\mu u_{ttxx}=-\alpha u_{xxxx}+\beta u_{txx}-\gamma u-\delta u_t+f(u)+\varepsilon h(x,t),\\
&u(0,x)=u_0(x),\quad u_t(0,x)=v_0(x),\\
& x\in(0,\pi),
\end{cases}
\end{equation}
where $\alpha=EI/\rho$ and $\mu=I_{\rho}/\rho$, for the constant mass density $\rho$, Young's modulus of elasticity $E$, the second moment of area of the beam's cross-section $I$ and the constant mass moment of inertia $I_{\rho}$. The parameter $\beta$ accounts for internal damping due to the elastic properties of the material, while the parameter $\gamma$ describes the linear stiffness of the foundation. The parameter $\delta$ accounts for linear external damping, e.g., due to dampers. The function $h$, which we assume to be periodic in time with frequency $\omega$, describes external forcing for some small parameter $\varepsilon$.

\begin{figure}[h]
	\begin{subfigure}{.5\textwidth}
		\centering
		\includegraphics[width=.8\linewidth]{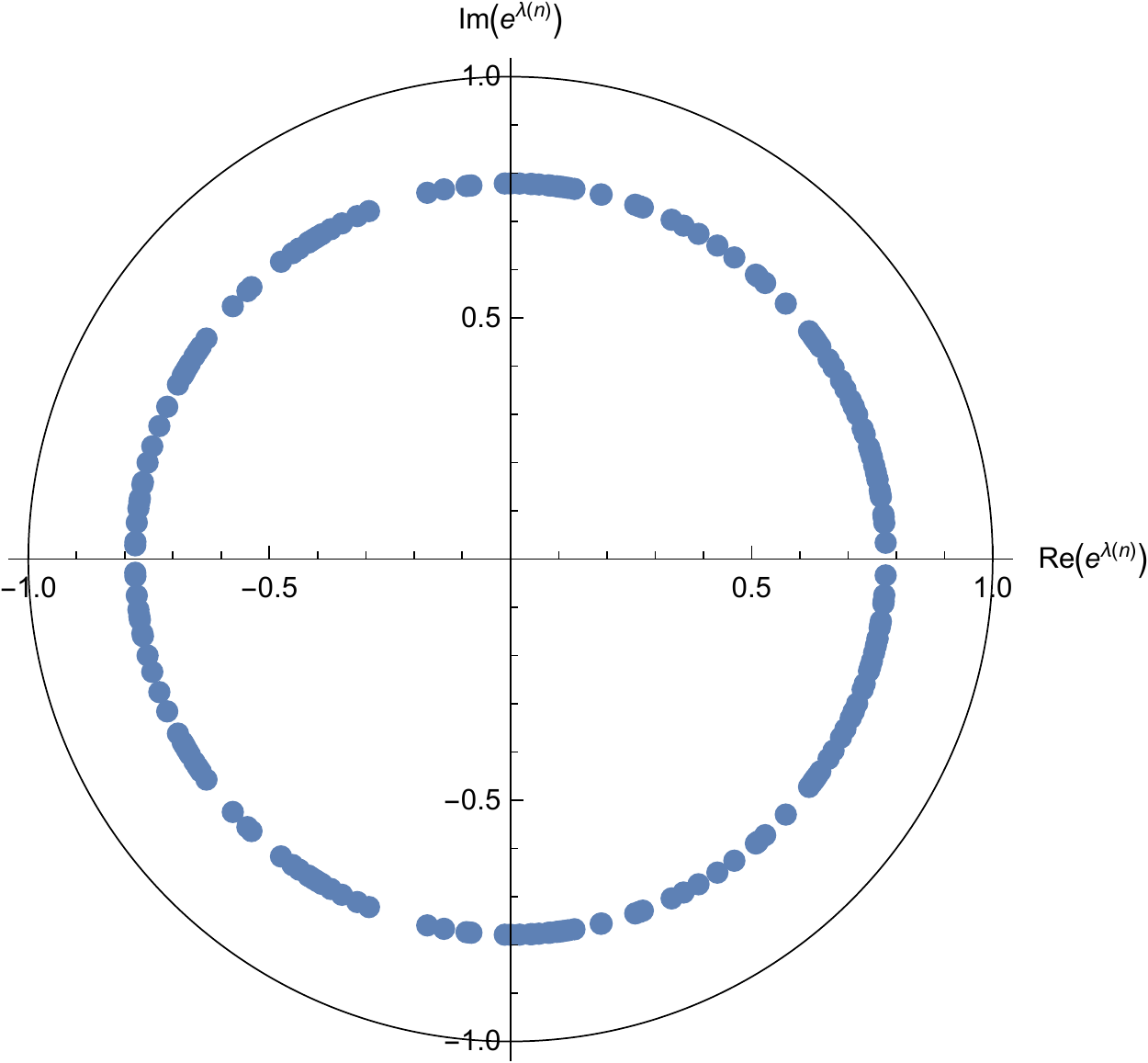}
		\caption{Pure viscous damping\\ ($\mu=\beta=0$)}
		\label{DampA}
	\end{subfigure}%
	\begin{subfigure}{.5\textwidth}
		\centering
		\includegraphics[width=.8\linewidth]{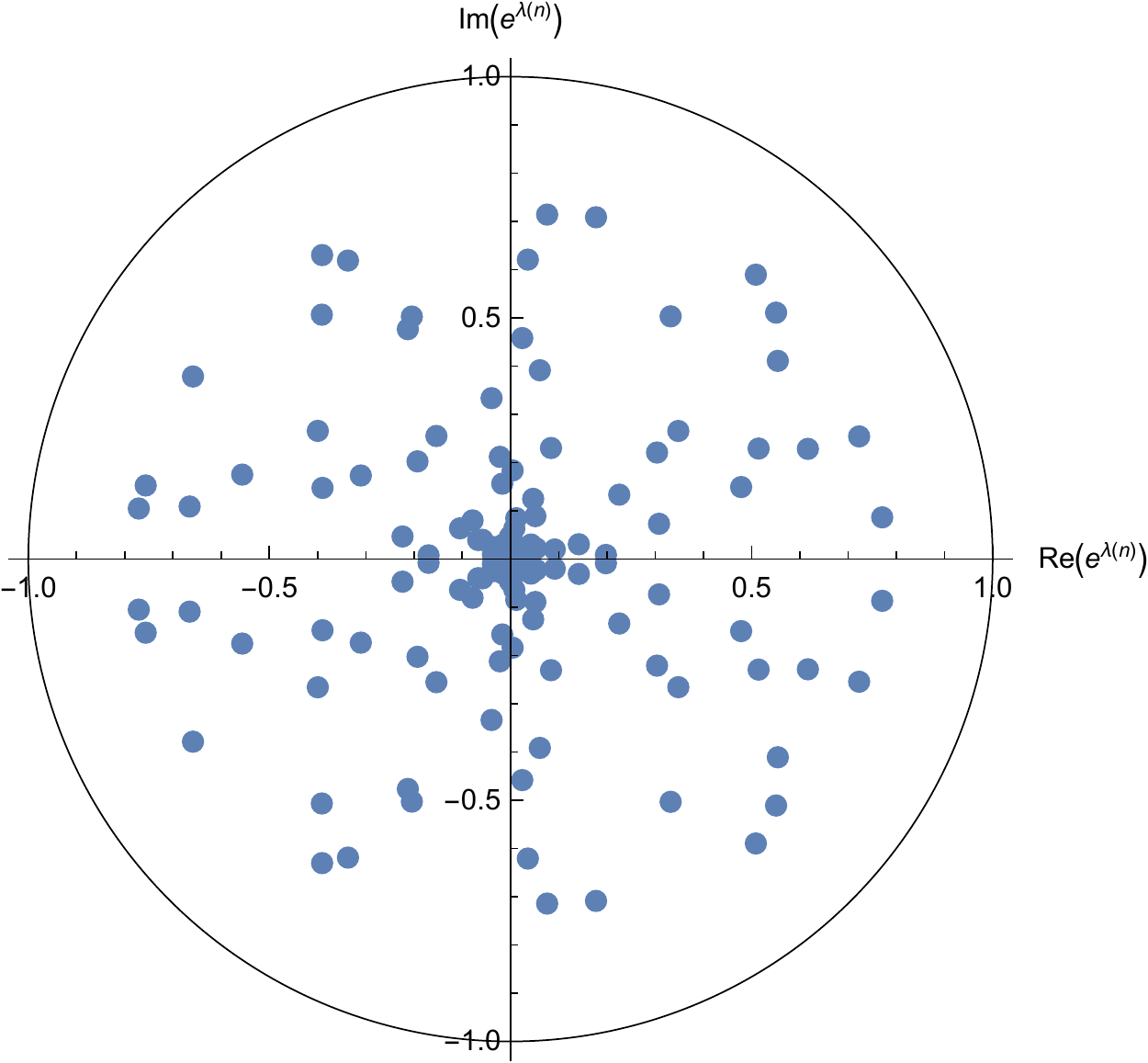}
		\caption{Pure visco-elastic damping\\ ($\mu=0$)}
		\label{DampB}
	\end{subfigure}
	\caption{Damping types excluded from our analysis}
	\label{Damp}
\end{figure}

\begin{remark}
(Our choice of damping)\\
We chose the set up of Rayleigh beam theory together with a damping proportional to the second derivative in space for two reasons.\\
First, if we considered Euler-Bernoulli beam theory with purely viscous damping, i.e. $\mu=\beta=0$ in (\ref{mainequ}), the real parts in the spectrum of the linearization of the right-hand side would be constant  (cf. (\ref{specA})). Consequently, the eigenvalues of the flow map of the unperturbed system would lie on a circle of radius $e^{-\frac{\delta}{2}}$ (cf. Figure \ref{DampA}), and there would be no dichotomy of decay rates for any choice of eigenspaces of the linearized flow map. As a consequence, even the linear system would not admit an obvious lower-dimensional subspace candidate for model reduction.\\
Second, we could have also considered Euler-Bernoulli beam theory with Kelvin-Voigt damping (visco-elastic damping; see, e.g., \cite{doi:10.1137/1.9781611970982.ch4}) by adding a term of the form $-\beta u_{txxxx}$ to the right-hand side of (\ref{mainequ}). In that case, however, the real parts in the spectrum of the linearization of the the right-hand side would decay like $n^4$, as indicated by (\ref{specA}). Therefore, the spectrum of the linearized flow map for the unperturbed system would accumulate at the origin as $n\to \infty$. This would imply that the flow map is not invertible at the trivial solution and therefore condition (1) of Assumption \ref{AsB} would be violated. This is because the $\partial^4/\partial x^4$- term would make the solution analytic for arbitrarily small times and the inverse problem would therefore be ill-posed. The same would hold true if we merely chose a damping term of the form $\beta u_{txx}$, i.e.  a damping term that is proportional to the bending rate  (cf. Figure \ref{DampB}). Beyond being a technical inconvenience, the non-invertability of the linearized PDE creates a conceptual issue, a conflict with the Newtonian principle of determinism for the mechanics of beam motion.\\
We remedy these technical difficulties by introducing a damping that is of the same order as the mass term that comes from Rayleigh beam theory. This essentially allows us to the treat the fourth-order problem (\ref{mainequ}) as a second-order problem. In particular, the spectrum is bounded and a family of slow eigenvalues exists, cf. Figure \ref{SpectrumRayleigh}. For a physical justification of our choice of damping, we again refer to \cite{doi:10.1137/1.9781611970982.ch4}, pages 135-140, especially page 139.
\begin{figure}[h]
	\centering
	\includegraphics[scale=0.6]{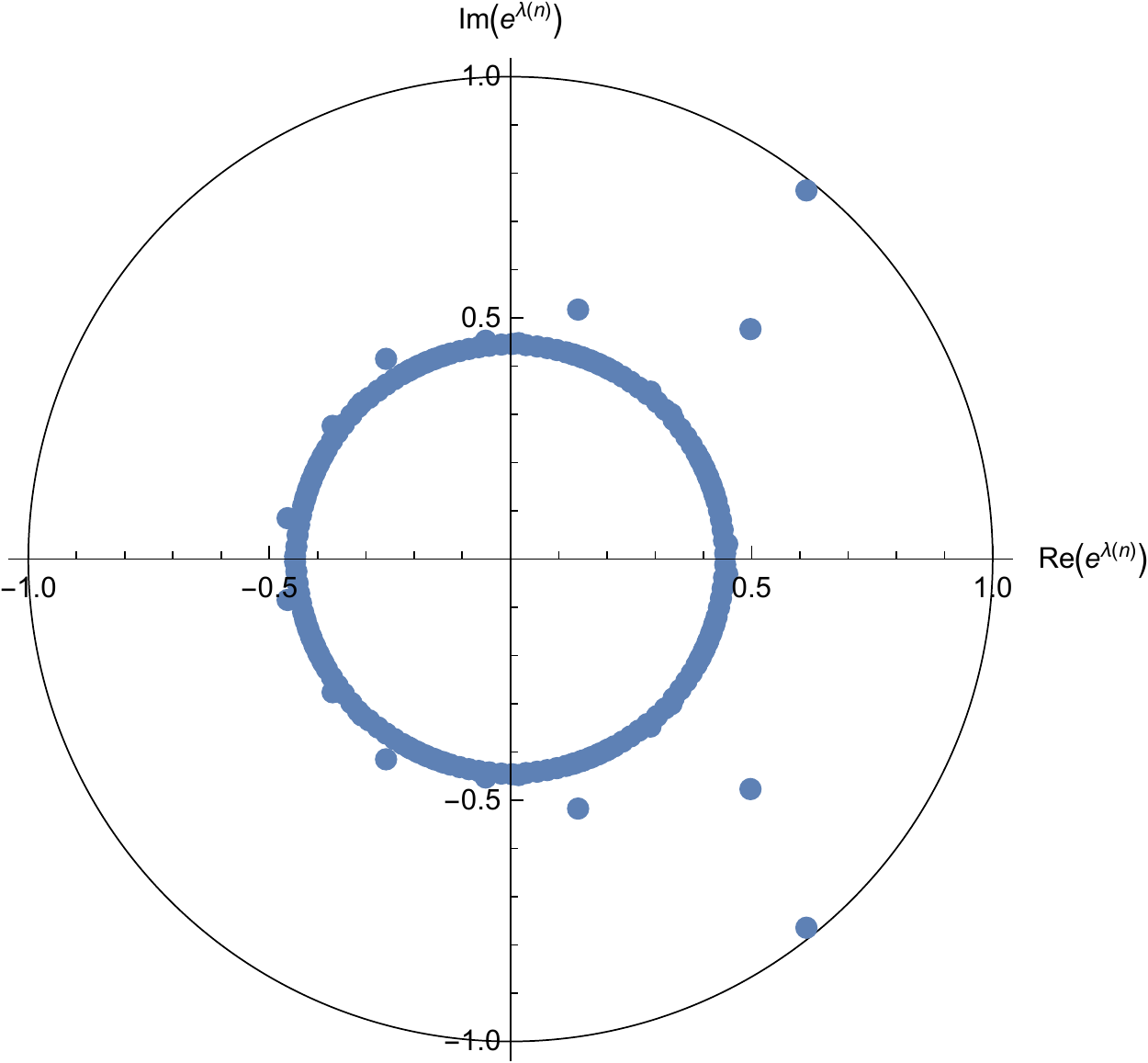}
	\caption{Typical spectrum of the linear part of Equation (\ref{mainequ}). The spectrum is inside the unit circle and the eigenvalues accumulate on a smaller circle.}
	\label{SpectrumRayleigh}
\end{figure}

\end{remark}

For the subsequent analysis, we have to make some further assumptions on the nonlinearity $f$ and the parameters in (\ref{mainequ}).\\

The function $f:\mathbb{R}\to\mathbb{R}$ describes the nonlinear external force interaction of the beam with the foundation. Since the linear part is already described by the parameter $\gamma$, it is reasonable to make the following general assumptions:
\begin{assumption}\label{AsA}
	The function $f:\mathbb{R}\to\mathbb{R}$ is of class $C^r$ for some $r\in\mathbb{N}\cup \{\infty,a\} $ and of polynomial growth 
	\begin{equation}\label{above1}
	|f(x)|\lesssim |x|^m
	\end{equation}
	for some $m>1$, as well as of polynomial growth in its derivative
	\begin{equation}\label{above2}
	|f'(x)|\lesssim |x|^{m-1},
	\end{equation}
	and satisfies
	\begin{equation}
	f(0)=0\quad,\quad f'(0)=0.
	\end{equation}
	The indefinite integral of $f:\mathbb{R}\to\mathbb{R}$ is a non-positive function 
	\begin{equation}\label{below}
	\int_0^xf(\xi)\, d\xi\leq 0,
	\end{equation}
	for all $x\in\mathbb{R}$.
\end{assumption}

\begin{remark}
	Assumptions (\ref{above1}) and (\ref{above2}) guarantee that the map $u\mapsto f(u)$ is sufficiently regular as a map on Banach spaces, while assumption (\ref{below}) ensures global existence of solutions to (\ref{mainequ}). This will become apparent in the subsequent paragraphs.	
\end{remark}
Setting $$S_{T}:=\mathbb{R}\mod T, $$ we will be making the following assumptions on the external, time-dependent forcing:
\begin{assumption}\label{Asforce}
	The function $h:(0,\pi)\times S_{\omega}\to\mathbb{R}$ is continuously differentiable as a map $h:\mathbb{R}\to L^2(0,\pi)$ and satisfies the same boundary conditions as the beam.
	In particular, we have that $\|h(x,t)\|_{L^2(0,\pi)}\leq  H_0$ and  $\|h_t(x,t)\|_{L^2(0,\pi)}\leq  H_1$ for any $t\in S_\omega$.
\end{assumption}
An example for a foundation that leads to an equation like (\ref{mainequ}) could be a bed of nonlinear springs. In \cite{SHAW1994319}, a foundation of cubic springs is considered where $f(u)=-C u^3$ for some $C>0$.\\
Requiring hinged ends (simply supported beam), we impose the following boundary conditions on the solution $u$:

\begin{equation}\label{BC1}
\begin{cases}
&u(t,0)=u(t,\pi)=0\\
&u_{xx}(t,0)=u_{xx}(t,\pi)=0.
\end{cases}
\end{equation}

\begin{remark}
	We could also require clamped ends by imposing the boundary condition\begin{equation}\begin{cases}&u(t,0)=u(t,\pi)=0\\&u_{x}(t,0)=u_{x}(t,\pi)=0,\end{cases}\end{equation} which would change the eigenbasis for the operator $\partial^4$ on $L^2(0,2\pi)$. The subsequent analysis could be carried out similarly, though.
\end{remark}

For a physical beam problem, the external damping parameter $\delta$, as well as the parameter $\beta$, accounting for internal damping, are small compared to the other parameters of the system. It is therefore reasonable to make the following set of assumptions:
\begin{assumption}\label{param}
	The non-negative parameters $\alpha,\beta,\gamma,\delta$ and $\mu$ in equation (\ref{mainequ}) satisfy
	\begin{enumerate}
		\item $\beta^2<4\alpha$,
		\item $2\beta\delta<4\gamma\mu$,
		\item $\delta^2<4\gamma$,
		\item $\delta\mu<\beta$.
	\end{enumerate}
\end{assumption}
\begin{remark}
	The meaning of the above assumptions will become apparent in the analysis of the linear spectrum of equation (\ref{mainequ}) (cf. Appendix I). The first three assumptions ensure that all eigenvalues of the linearization have non-zero imaginary parts. This means that the beam equation is close to a conservative system and hence there are no overdamped modes. The conditions depend on a sign criterion for a third order polynomial and can undoubtedly be refined.\\
	The fourth inequality in Assumption \ref{param} guarantees that the real parts of the eigenvalues of the linearization decrease with the frequency. This will permit us to extract a unique attracting slow manifold from the dynamics near the equilibrium configuration of the forced beam. If the inequality in the fourth assumption is reversed, it is possible to extract a unique attracting fast manifold.
\end{remark}

\section{Notation and Basic Definitions}

Let $f,g:\mathbb{R}\to\mathbb{R}$ be two functions. We write
$$f(x)\lesssim g(x) \iff f(x)\leq C g(x)$$
for some $C>0$. Similarly for the symbol $``\gtrsim``$.\\
Let $X$ be a real or complex Banach space and let $U\subset X$ be an open set. For a real or complex Banach space $Y$, let $C^r(U,Y)$ be the space of $r$-times continuously differentiable functions from $U$ to $Y$. The space $C^{\infty}(U,Y)$ consists of all functions $f:U\to Y$, which belong to $C^r(U,Y)$ for all $r\geq 1$, while the space $C^a(U,Y)$ consists of all that are locally analytic in $U$. The spaces $C^r(U,Y)$ for some $r\geq 1$ and the space $C^a(U,Y)$ are Banach spaces, while the space $C^\infty(U,Y)$ is a Fr\'echet space.\\
For any positive number $1\leq p<\infty$, let $L^p(0,\pi)$ denote the space of complex-valued, $p$-integrable functions on the interval $(0,\pi)$. Namely, $L^p(0,\pi)$ consists of all functions $f:(0,\pi)\to\mathbb{C}$ such that
$$\|f\|_{L^p}=\left(\int_0^\pi|f(x)|^p\, dx\right)^{\frac{1}{p}}<\infty.$$
The space $L^\infty(0,\pi)$ consists of all bounded functions on $[0,\pi]$, admitting the norm $$\|f\|_{L^\infty(0,\pi)}=\sup_{x\in [0,\pi]}|f(x)|.$$
Of particular interest for our analysis is the space $L^2(0,\pi)$, which may be endowed with a Hilbert space structure via the inner product $$\langle f,g\rangle_{L^2}=\frac{1}{\pi}\int_0^\pi f(x)g(x)^{*}\, dx,$$
with the star denoting complex conjugation. Any function $f\in L^2(0,\pi)$ can then be written uniquely as a Fourier series
$$f(x)=\sum_{n\in\mathbb{Z}}\hat{f}_n e^{2\ri n x},$$
where $$\hat{f}_n=\left\langle f,e^{2\ri n x}\right\rangle_{L^2}=\frac{1}{\pi}\int_0^\pi f(x)e^{-2\ri n x}\, dx$$
is the $n$-th Fourier coefficient of $f$.
By Parseval's theorem, cf. \cite{adams2003sobolev}, the $L^2$-inner product may be expressed as
\begin{equation}\label{ParsI}
\langle f,g\rangle_{L^2}=\sum_{n\in\mathbb{Z}}\hat{f}_n\hat{g}_n^{*},
\end{equation}
for the Fourier coefficients $\hat{f}_n$ of $f$ and $\hat{g}_n$ of $g$, respectively, and the $L^2$-norm of $f$ can be written as 
\begin{equation}\label{ParsN}
\|f\|_{L^2}^2= \sum_{n\in\mathbb{Z}}|\hat{f}_n|^2.
\end{equation}
For any $s\in\mathbb{N}$, let $H^s(0,\pi)$ be the space of functions $f\in L^2(0,\pi)$ such that \begin{equation}\label{defHs}
\|f\|_{H^s}^2=\sum_{n\in\mathbb{Z}}\left(\frac{1-(2n)^{2s+2}}{1-(2n)^2}\right)|\hat{f}_n|^2<\infty.
\end{equation}
The Hilbert spaces $H^s(0,\pi)$ are called \textit{Sobolev space of order $s$}. For a detailed theory of Sobolev spaces and related topics, we refer to \cite{adams2003sobolev}. 
\begin{remark}
	We chose to define the Sobolev space in terms of the decay of Fourier coefficients. Equivalently, one can characterize elements of a Sobolev space by summability properties. In fact, any element $f\in H^s(0,\pi)$ possesses weak derivatives up to order $s$ in $L^2(0,\pi)$. This enables us to write the norm in $H^s$ equivalently in an integral form as
	\begin{equation}\label{HsInt}
	\|f\|_{H^s}^2=\frac{1}{\pi}\int_0^\pi\sum_{j=1}^s\left(\frac{d^{j} f}{dx^{j}}(x)\right)^2\, dx,
	\end{equation}
	where the derivatives have to be understood in the weak sense. For details, we refer to \cite{adams2003sobolev}.
\end{remark}
In order to incorporate the boundary condition (\ref{BC1}), we introduce the following linear subspace of $H^s(0,\pi)$ for even $s\in\mathbb{N}$:
\begin{equation}\label{Hs0}
H^s_0(0,\pi)=\left\{f\in H^s(0,\pi):\frac{d^{2j} f}{dx^{2j}}(0)=\frac{d^{2j} f}{dx^{2j}}(\pi)=0,\quad\text{for }j=1,...,\frac{s}{2}\right\}.
\end{equation}
Any element $f\in H^s_0(0,\pi)$ may be written as a Fourier-sine series 
$$f(x)=\sum_{n=1}^{\infty}\hat{f}_n\sin(nx).$$
	
A possible unbounded operator $A:\mathcal{H}\to\mathcal{H}$ on a Hilbert space with domain $\mathcal{D}(A)$ is called \textit{dissipative} if for all $x\in\mathcal{D}(A)$ we have that
\begin{equation}\label{dissHi}
\langle Ax,x\rangle\leq 0
\end{equation}
For a possibly unbounded operator $A:\mathcal{H}\to\mathcal{H}$ on a Hilbert space with domain $\mathcal{D}(A)$, let $\sigma(A)$ denote its spectrum and let $\rho(A)$ denote its resolvent set.
The \textit{point spectrum} of $A$, denoted $\sigma_p(A)$, then consists of all eigenvalues of finite multiplicity.For a detailed discussion of spectral properties of unbounded operators on Hilbert spaces and related topics, we refer to \cite{hislop2012introduction} and \cite{teschlmathematical}.\\

\section{Spectral Submanifolds}

We will be interested in the behavior of solutions to equation (\ref{mainequ}) in the neighborhood of special solutions. In the case of no external forcing, the special solution of interest will be the fixed point $u\equiv 0$, while in the presence of external periodic forcing, the special solution will be time-periodic. Existence of such a time-periodic solution is proved in Appendix I, Lemma \ref{fixed}, by a Poincar\'e map argument. Following \cite{Haller2016}, we call all solutions with a finite number of frequencies \textit{nonlinear normal modes}.\\
To study the beam equation (\ref{mainequ}), we rewrite it as a first order system:
\begin{equation}\label{mainequV}
\left(\begin{array}{c}u_t\\ v_t\end{array}\right)=A\left(\begin{array}{c}u\\ v\end{array}\right)+\left(\begin{array}{c}0\\f(u)+\varepsilon h(x,t)\end{array}\right),
\end{equation}
where we have set $v:=u_t$ and
\begin{equation}\label{defMA}
A:=\left(\begin{matrix}
0 & 1 \\ -\mathcal{M}^{-1}(\alpha\partial^4+\gamma) & \mathcal{M}^{-1}(\beta\partial^2-\delta)\end{matrix}\right),
\end{equation}
for the operator
\begin{equation}
\mathcal{M}=1-\mu\partial^{2}.
\end{equation}
The spectrum of the matrix of operators (\ref{defMA}) is given explicitly by
\begin{equation}\label{specA}
\begin{split}
\sigma(A)&=\{\lambda_n^{\pm}\}_{n\in\mathbb{N}^{+}}\\
&=\left\{-\frac{\beta n^2+\delta}{2+2\mu n^2}\pm i \sqrt{\frac{\alpha n^4+\gamma}{1+\mu n^2}-\left(\frac{\beta n^2+\delta}{2+2\mu n^2}\right)^2}\right\}_{n\in\mathbb{N}^{+}}.
\end{split}
\end{equation}
By Assumption \ref{param}, the real parts of the eigenvalues in (\ref{specA}) are negative, monotonically decreasing and bounded from below (cf. Figure \ref{PlotRealPartSpectrum}).
\begin{figure}
	\centering
	\includegraphics[scale=0.7]{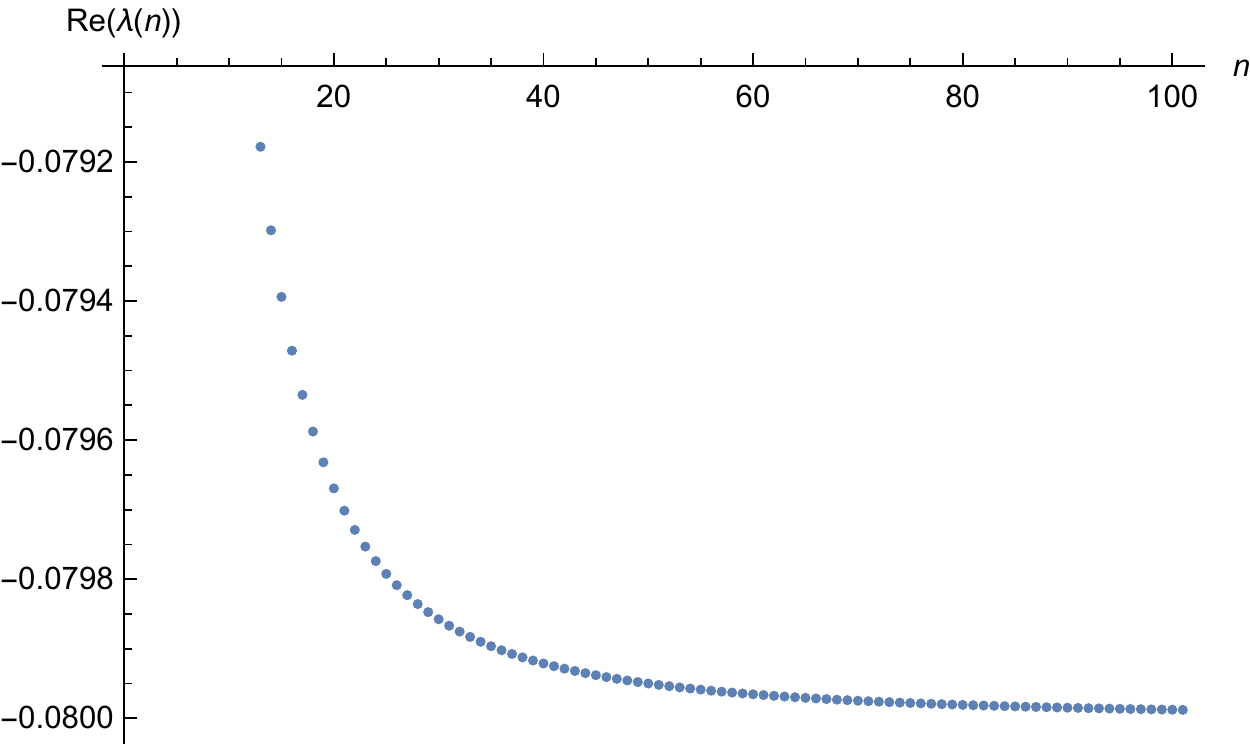}
	\caption{The real parts of the eigenvalues given by formula (\ref{specA}), for parameter values $\beta=0.08$, $\delta=0.04$ and $\mu=0.5$. Observe that the real parts converge to $-0.08$ as $n\to\infty$.
}
	\label{PlotRealPartSpectrum}
\end{figure}
The corresponding eigenvectors are given by
\begin{equation}\label{eigvec}
\{v_n^+,v_n^-\}=\left\{\left(\begin{array}{c}
1\\\lambda_{n}^+
\end{array}\right)\sin(nx),\left(\begin{array}{c}
1\\\lambda_{n}^-
\end{array}\right)\sin(nx)\right\}.
\end{equation}

As the underlying space for equation (\ref{mainequ}), we choose the standard energy space $\mathcal{H}:=H^1_0(0,\pi)\times L^2(0,\pi)$. We prove well-posedness of equation (\ref{mainequV}) in the Hilbert space $\mathcal{H}$ in Appendix 1, Theorem \ref{semi} and Theorem \ref{ExUn}, while we prove global existence of solutions to equation (\ref{mainequV}) in Appendix I, Proposition \ref{global}. This permits us to define a semi-flow map $\phi_\varepsilon^T: \mathcal{H}\to\mathcal{H}, (u_0,v_0)\mapsto \phi_\varepsilon^T(u_0,v_0)$ for any fixed time $T>0$ and for $\varepsilon>0$, which maps initial conditions $(u_0,v_0)$ to solutions of equation (\ref{mainequV}). Let $\mathcal{A}_\varepsilon$ denote the linearization of the flow map at the fixed time $T>0$.

Following \cite{Haller2016}, introduce the following definition of spectral submanifolds for the flow map of equation (\ref{mainequV}).
\begin{definition}
A \textit{spectral submanifold} (SSM) associated with a spectral subspace $\mathcal{E}_\varepsilon$ of the operator $\mathcal{A}_\varepsilon$ around the special solution $U_0^\varepsilon$ is a manifold, denoted by $W(\mathcal{E}_\varepsilon)$, with the following properties:
\begin{enumerate}
	\item The manifold $W(\mathcal{E})_\varepsilon$ is forward-invariant under the flow map $\phi_\varepsilon^T$,  tangent to $\mathcal{E}_\varepsilon$ at $U_0^\varepsilon$ and has the same dimension as $\mathcal{E}_\varepsilon$.
	\item The manifold $W(\mathcal{E}_\varepsilon)$ is strictly smoother than any other manifold satisfying (1).
\end{enumerate}
A \textit{slow spectral submanifold} (slow SSM) is an SSM $W(\mathcal{E}_{\varepsilon, slow})$ associated with a spectral subspace $\mathcal{E}_{\varepsilon, slow}$ of finitely many eigenvalues with the largest real parts within the total spectrum of $\mathcal{A}_\varepsilon$.
\end{definition}
In the presence of time-periodic forcing, the fixed point $U_0^\varepsilon$ of $\phi_\varepsilon^T$ corresponds to a nontrivial periodic orbit of \ref{mainequ}, while in the case of no external forcing, $U_0^0=0$ also determines a fixed point for the flow map of equation \ref{mainequ}. For model reduction purposes, the dynamics on a slow SSM is a faithful approximation to the full dynamics for large times. For a discussion on an SSM that is not associated with the slowest modes, we refer to Appendix II.\\
We choose a $2N$-dimensional eigenspace $\mathcal{E}\cong\mathbb{C}^{2N}$ of the linear operator $A$ with corresponding eigenvalues $\lambda_1,...,\lambda_{N}$, $N\in\mathbb{N}^+$. For such an eigenspace, we define the \textit{relative spectral quotient} as the positive integer
\begin{equation}\label{quotient}
q(\mathcal{E}):=\left[\frac{\inf_{j>  N}\text{Re}\lambda_{j}}{\text{Re}\lambda_{1}}\right]\in\mathbb{N}^+,
\end{equation}
where the square bracket denotes the integer part.\\
For the unperturbed equation, the relative spectral quotient is a measure of how unique the SSM $\mathcal{W}(\mathcal{E})$ is. Specifically, $\mathcal{W}(\mathcal{E})$ will turn out to be unique among class $q(\mathcal{E})$+1 invariant manifolds tangent to $\mathcal{E}$. Note that in the presence of forcing, the quantity (\ref{quotient}) implicitly depends upon $\varepsilon$ (cf. Appendix II, in particular Lemma \ref{perspec}).\\
In order to apply an abstract existence result for SSMs by Cabr\'{e} et al. \cite{Cab2003}, we make the following non-resonance assumption on the spectrum of the matrix of operators $A$.

\begin{assumption}\label{nonresmain}
Let $\lambda_1,..,\lambda_N$ be the $N$ eigenvalues of $A$ with the $N$ largest real parts and let $q(\mathcal{E})$ be the corresponding relative spectral quotient. Assume that
	\begin{equation}
	s_1\lambda_1+s_2\lambda_2+...+s_N\lambda_N\neq \lambda_j,\quad s_i\in\mathbb{Z}^{+},
	\end{equation}
for $2\leq s_1+s_2+...+s_N\leq q(\mathcal{E})$ and $j> N$.
\end{assumption}

\begin{theorem}\label{mainthmdifeps}
Assume that the parameters in equation (\ref{mainequ}) are such that Assumption \ref{param} and Assumption \ref{nonresmain} are satisfied. Assume further that the nonlinearity $f(u)$ satisfies Assumption \ref{AsA} with $r=a$, while the external forcing satisfies Assumption \ref{Asforce}. Then, for sufficiently small $\varepsilon>0$, there exists a unique, analytic SSM, $W(\mathcal{E}_\varepsilon)$, that is tangent to the spectral subspace $\mathcal{E}_\varepsilon$ along the periodic solution $U_0^\varepsilon$.\\
\renewcommand{\theenumi}{\roman{enumi}}
\begin{enumerate}
	\item The manifold $\mathcal{W}(\mathcal{E}_\varepsilon)$ is unique among all class $C^{q(\mathcal{E}_\varepsilon)+1}$ forward-invariant manifolds tangent to $\mathcal{E}_\varepsilon$ along $U_0^\varepsilon$.
	\item The dynamics on the SSM are conjugate to a differential equation with a polynomial right-hand side of degree not larger than $q(\mathcal{E}_\varepsilon)-1$.
\end{enumerate}
\end{theorem}
\begin{proof}
We refer to Appendix II.
\end{proof}
Theorem \ref{mainthmdifeps} states that there exists an open subset $U\subset\mathbb{C}^N$ and a parameterization $K_{\varepsilon}:U\to\mathcal{H}$ such that the slow SSM associated to the spectral subspace $\mathcal{E}_\varepsilon$ is given by $W(\mathcal{E}_{\varepsilon})=K_{\varepsilon}(U)$. Furthermore, there exists a polynomial map $R_\varepsilon:U\to U$ of degree not larger than $q(\mathcal{E})-1$ such that the equation
\begin{equation}\label{mainthmdifeqeps}
A\cdot K_\varepsilon+G_\varepsilon\circ K_\varepsilon =D K_\varepsilon\cdot R_\varepsilon+\omega D_\theta K
\end{equation}
holds true.\\
In the case of no external forcing, we obtain a similar result under slightly weaker assumptions on the nonlinearity.
\begin{theorem}\label{mainthmdif}
Let $\varepsilon=0$ and assume that  that the parameters in equation (\ref{mainequ}) are such that Assumption \ref{param} and Assumption (\ref{nonresmain}) are satisfied. Assume further that the nonlinearity satisfies Assumption \ref{AsA} with $r\in\mathbb{N}\cup \{\infty,a\} $ and $q(\mathcal{E})\leq r$. Then there exists a unique SSM $W(\mathcal{E})$ of class $C^r$ that is tangent to the spectral subspace $\mathcal{E}$ at the trivial solution $U_0=0$.\\
\renewcommand{\theenumi}{\roman{enumi}}
\begin{enumerate}
	\item The manifold $\mathcal{W}(\mathcal{E})$ is unique among all class $C^{q(\mathcal{E})+1}$ forward-invariant manifolds tangent to $\mathcal{E}$ along $U_0$.
	\item The dynamics on the SSM are conjugate to a differential equation with a polynomial right-hand side of degree not larger than $q(\mathcal{E})-1$.
\end{enumerate}

\end{theorem}
\begin{proof}
We refer to Appendix II.
\end{proof}
There exists an open subset $U\subset\mathbb{C}^N$ and a parameterization $K:U\to\mathcal{H}$ such that the slow SSM associated to the spectral subspace $\mathcal{E}$ is given by $W(\mathcal{E})=K(U)$. Finally, there exists a polynomial map $R:U\to U$ ot degree not larger than $q(\mathcal{E})-1$ such that the equation
\begin{equation}\label{mainthmdifeq}
A\cdot K+G_{0}\circ K=DK\cdot R
\end{equation}
holds true.

\subsection{An example with no external forcing} To illustrate the above results, we consider an example with no external, time-dependent forcing present, i.e., set $\varepsilon=0$. In the following, we will compute a reduction of equation (\ref{mainequ}) to a two-dimensional, slow invariant manifold around the fixed point $u=0$.
Consider a cubic nonlinearity
\begin{equation}
f(u)=-\kappa u^3,
\end{equation}
and set the internal damping $$\beta=\frac{4\delta\mu}{1-3\mu}.$$
The parameters are chosen in a way that facilitates the calculations, and hence do not correspond to a particular physical beam geometry. Evaluating the spectral quotient (\ref{quotient}), we find that
\begin{equation}
q=\left[\frac{\inf_{n}\text{Re}\lambda_n}{\text{Re}\lambda_1}\right]= \left[\frac{\beta(1+\mu)}{(\beta+\delta)\mu}\right]=4,
\end{equation} 
meaning that we are looking for the analytic invariant manifold around the fixed point $u\equiv 0$ that is unique amongst all $C^5$ invariant manifolds. The reduced dynamics are given by a polynomial right-hand side of degree not larger than $3$. As the parametrization space for the invariant manifold, we choose the eigenspace associated with the complex eigenvalues
\begin{equation}\label{l1}
\lambda_1=-\frac{\beta +\delta}{2+2\mu }+i\sqrt{\frac{\alpha +\gamma}{1+\mu }-\left(\frac{\beta +\delta}{2+2\mu }\right)^2}.
\end{equation}
That is to say, we choose as $\mathcal{E}$ the spectral subspace
\begin{equation}\label{specun}
\mathcal{E}=\spn\left\{\left(\begin{array}{c}1\\ \lambda_1\end{array}\right)\sin(x),\left(\begin{array}{c}1\\ \overline{\lambda}_1\end{array}\right)\sin(x)\right\}\cong\mathbb{C}^2.
\end{equation}
We will denote coordinates in the space $\mathcal{E}$ as $(z,\overline{z})$.

In accordance with Theorem \ref{mainthm}, the dynamics on the invariant manifold $\mathcal{W}(\mathcal{E})$ can be approximated by a Taylor-series expansion on the space $\mathcal{E}$. That is, there is an analytic function $K:\mathbb{C}^2\to H^1_0(0,\pi)\times L^2(0,\pi)$ given by
\begin{equation}\label{defK}
K(z)=\sum_{|n|=1}^\infty K_n z^n,
\end{equation}
for $z=(z_1,z_2)$, $K_n\in H^1_0(0,2\pi)\times L^2(0,\pi)$ and $n=(n_1,n_2)\in\mathbb{N}^2$.\\
Note that the function $K$ parametrizes the unique slow manifold in the space $ H^1_0(0,2\pi)\times L^2(0,\pi)$, which also admits complex Fourier coefficients. Of course, we have again that $z_2=\overline{z_1}$ and therefore, one finds that the $K_n$'s written in the eigenbasis (\ref{eigvec}) satisfy the relation
\begin{equation}\label{symK}
K_{(n_2,n_1)}=\left(\begin{matrix}
0 &1\\ 1 & 0
\end{matrix}\right)\overline{K}_{(n_1,n_2)},
\end{equation}
as a comparison of powers in $z$ and $\overline{z}$ shows.

Since we have assumed that the damping is small, the slow modes are almost in resonance:
\begin{equation}\label{nearres}
2\lambda_1+\overline{\lambda_1}\approx \lambda_1,\quad \lambda_1+2\overline{\lambda_1}\approx \overline{\lambda_1}.
\end{equation}
This follows from the explicit formula for the eigenvalues (\ref{specA}) and the smallness of the parameters $\mu$ and $\delta$. A similar line of reasoning has been employed in \cite{Szalai2016} for the computation of backbone curves.

In view of Theorem \ref{mainthmdif}, we assume that the conjugated dynamics on the slow SSM $\mathcal{W}(\mathcal{E})$ are given by
\begin{equation}\label{defR}
\left(\begin{array}{c}
\dot{z}\\ \dot{\overline{z}}
\end{array}\right)=\left(\begin{array}{c}\lambda_1 z+R_0 z^2\overline{z}\\ \overline{\lambda}_1 \overline{z}+\overline{R}_0 z\overline{z}^2
\end{array}\right).
\end{equation}
Note that the tangency to the spectral subspace $\mathcal{E}$, as it is guaranteed by Theorem \ref{mainthmdif}, is incorporated in the linear part of $R$. This specific form of $R$ will turn out to be the simplest reduced dynamics on the SSM, while at the same the coefficients of the parametrization $K$ do not include small denominators for an appropriate choice of $R_0$.

Inserting the defining equations (\ref{defK}) and (\ref{defR}) for $K$ and $R$, respectively, into equation (\ref{mainthmdifeq}), we obtain the following relation:
\begin{equation*}
\begin{split}
&\sum_{|n|=1}^{\infty} A\cdot K_{n}z^{n} +\mathbf{f}(K(z))=\\
&=\sum_{|n|=1}^{\infty}\Big((\lambda_1 z_1+R_0 z_1^2z_2)n_1z_2+(\overline{\lambda}_1 z_2+\overline{R}_0 z_1z_2^2)n_2z_1\Big)K_{(n_1,n_2)}z_1^{n_1-1}z_2^{n_2-1}\\
&=\sum_{|n|=1}^{\infty}(\lambda_1n_1+\overline{\lambda}_1n_2)K_nz^n+\sum_{|n|=1}^{\infty}(R_0n_1+\overline{R}_0n_2)K_nz^{n+(1,1)}\\
&= \sum_{|n|=1}^{\infty}(\lambda_1n_1+\overline{\lambda}_1n_2)K_nz^n+\sum_{|n|=3}^{\infty}(R_0(n_1-1)+\overline{R}_0(n_2-1))K_{n-(1,1)}z^{n}.
\end{split}
\end{equation*}
Since $f(K(z))=\mathcal{O}(|z|^3)$, we find that equation (\ref{mainthmdif}), at order one, becomes
\begin{equation}
A\cdot K_{(1,0)}=\lambda_1 K_{(1,0)},\quad A\cdot K_{(0,1)}=\overline{\lambda}_1 K_{(0,1)},
\end{equation}
which implies that $K_{(1,0)}$ and $K_{(0,1)}$ are eigenvectors for the operator $A$ with eigenvalues $\lambda_1$ and $\overline{\lambda}_1$, respectively, i.e.,
\begin{equation}\label{Kcoef}
 K_{(1,0)}=\left(\begin{array}{c}
1\\0
\end{array}\right)\sin(x),\quad K_{(0,1)}=\left(\begin{array}{c}
0\\1
\end{array}\right)\sin(x).
\end{equation}
At order two, equation (\ref{mainthmdif}) becomes
\begin{equation}\label{Ko2}
\begin{split}
&A\cdot K_{(2,0)}=2\lambda_1 K_{(2,0)},\\
&A\cdot K_{(1,1)}=(\lambda_1+\overline{\lambda}_1)K_{(2,0)},\\
&A\cdot K_{(0,2)}=2\overline{\lambda}_1, K_{(2,0)}.
\end{split}
\end{equation}
Since the $K_n$'s are written in the eigenbasis of of the operator $A$, the only possible solution to equation (\ref{Ko2}) is
\begin{equation}\label{Ko2sol}
 K_{(2,0)}=0,\quad  K_{(1,1)}=0,\quad  K_{(0,2)}=0.
\end{equation}
Knowing the first order terms, we may now easily represent the cubic nonlinearity in the eigenbasis of the operator $A$. For a vector $\mathbf{x}=(x_1,x_2)$, we denote by $[\mathbf{x}]_i=x_i$, i=1,2, the $i$-th component of $\mathbf{x}$. The non linearity $f(u)=\kappa u^3$ then takes the form
\begin{equation}\label{nonlin}
\begin{split}
G(K(z))=&\mathbf{V}^{-1}\cdot\left(\begin{array}{c} 0\\
-\kappa \left(\sum_{|n|=1}^\infty [\mathbf{V}\cdot K_n]_1 z^n\right)^3
\end{array}\right)\\[0.5cm]
&= \frac{\kappa}{4} (z+\overline{z})^3\left(\frac{3\sin(x)}{\overline{\lambda}_1-\lambda_1}-\frac{\sin(3x)}{\overline{\lambda}_3-\lambda_3}\right)\left(\begin{array}{c}
1\\ -1
\end{array}\right)+\mathcal{O}(|z|^4),
\end{split}
\end{equation}
where we have set the operator
\begin{equation}
\mathbf{V}:  H^1_0(0,\pi)\times L^2(0,\pi)\to H^1_0(0,\pi)\times L^2(0,\pi)
\end{equation} as
\begin{equation}\label{defV}
\mathbf{V}\cdot \sum_{m=1}^{\infty} \left(\begin{array}{c}
u_m\\ v_m 
\end{array}\right)\sin(mx):=\sum_{m=1}^{\infty}  \left(\begin{matrix}
1 & 1\\ \lambda_m & \overline{\lambda}_m
\end{matrix}\right)\left(\begin{array}{c}
u_m \\ v_m\end{array}\right)\sin(mx),
\end{equation}
which realizes the change of basis from physical coordinates to coordinates in the eigenbasis of $A$.  We may now compute the terms of order three in the expansion of $K$ by comparison of powers in $z$ and $\overline{z}$. Here, we only show the computations for the coefficient $K_{(3,0)}$ and $K_{(2,1)}$. The equations for the other coefficients then follow from the symmetry condition in equation (\ref{symK}). Using equation (\ref{nonlin}) and the eigenexpansion of the parametrization $K$, we find at order $z_1^3$:
\begin{equation}
A\cdot K_{(3,0)}+\frac{\kappa}{4} \left(\frac{3\sin(x)}{\overline{\lambda}_1-\lambda_1}-\frac{\sin(3x)}{\overline{\lambda}_3-\lambda_3}\right)\left(\begin{array}{c}
1\\ -1
\end{array}\right)=3\lambda_1 K_{(3,0)},
\end{equation}
which implies that
\begin{equation}
K_{(3,0)}=\frac{3\kappa\sin(x)}{4(\overline{\lambda}_1-\lambda_1)}\left(
\frac{1}{2\lambda_1}, \frac{1}{\overline{\lambda}_1-3\lambda_1}
\right)+\frac{\kappa\sin(3x)}{4(\overline{\lambda}_3-\lambda_3)}\left(
\frac{1}{\lambda_3-3\lambda_1}, \frac{1}{3\lambda_1-\overline{\lambda}_3}
\right).
\end{equation}
Similarly, at order $z_1^2z_2$:
\begin{equation}
A\cdot K_{(2,1)}+\frac{3\kappa}{4} \left(\frac{3\sin(x)}{\overline{\lambda}_1-\lambda_1}-\frac{\sin(3x)}{\overline{\lambda}_3-\lambda_3}\right)\left(\begin{array}{c}
1\\ -1
\end{array}\right)=(2\lambda_1+\overline{\lambda}_1)K_{(2,1)}+R_0K_{(1,0)},
\end{equation}
which has the solution
\begin{equation}\label{Kcoef21}
\begin{split}
 K_{(2,1)}=&\frac{\sin(x)}{4(\overline{\lambda}_1-\lambda_1)}\left(
\frac{9\kappa-4(\overline{\lambda}_1-\lambda_1)R_0}{\lambda_1+\overline{\lambda}_1}, -\frac{9\kappa}{2\lambda_1}
\right)\\
&+\frac{3\kappa\sin(3x)}{4(\overline{\lambda}_3-\lambda_3)}\left(
\frac{1}{\lambda_3-2\lambda_1-\overline{\lambda}_1}, \frac{1}{\overline{\lambda}_1+2{\lambda}_1-\overline{\lambda}_3}
\right).
\end{split}
\end{equation}
By equation (\ref{symK}), the remaining two coefficients are given by
\begin{equation}\label{Kcoef3}
\begin{split}
K_{(1,2)}=&\frac{\sin(x)}{4(\overline{\lambda}_1-\lambda_1)}\left(\frac{9\kappa}{2\overline{\lambda}_1},
\frac{4(\lambda_1-\overline{\lambda}_1)\overline{R}_0-9\kappa}{\lambda_1+\overline{\lambda}_1} 
\right),\\
&+\frac{3\kappa\sin(3x)}{4(\overline{\lambda}_3-\lambda_3)}\left(
\frac{1}{\lambda_3-2\overline{\lambda}_1-\lambda_1}, \frac{1}{\lambda_1+2\overline{\lambda}_1-\overline{\lambda}_3}\right),\\[0.3cm]
 K_{(0,3)}=&\frac{3\kappa\sin(x)}{4(\overline{\lambda}_1-\lambda_1)}\left( \frac{1}{3\overline{\lambda}_1-\lambda_1},
 -\frac{1}{2\overline{\lambda}_1}
 \right)+\frac{\kappa\sin(3x)}{4(\overline{\lambda}_3-\lambda_3)}\left(
 \frac{1}{\lambda_3-3\overline{\lambda}_1}, \frac{1}{3\overline{\lambda}_1-\overline{\lambda}_3}
 \right).
\end{split}
\end{equation}
Since $\lambda_1+\overline{\lambda}_1\approx 0$ by the smallness assumption on the damping, (\ref{nearres}), the quantities $K_{(2,1)}$ and $K_{(1,2)}$ in (\ref{Kcoef}) would contain large denominators. This, however, would limit the validity of the Taylor series expansion (\ref{defK}) to a smaller domain, which is unfavorable in applications. We, therefore, set the parameter $R_0$ in (\ref{defR}) as
\begin{equation}
R_0=\frac{9\kappa}{4(\overline{\lambda}_1-\lambda_1)}=\frac{9\kappa\ri}{8\text{ Im }\lambda_1},
\end{equation}
which eliminates the small denominators in equations (\ref{Kcoef21}) and (\ref{Kcoef3}):
\begin{equation}
\begin{split}
 K_{(2,1)}:=&\frac{\sin(x)}{4(\overline{\lambda}_1-\lambda_1)}\left(
0, -\frac{9\kappa}{2\lambda_1}
\right)\\
&+\frac{3\kappa\sin(3x)}{4(\overline{\lambda}_3-\lambda_3)}\left(
\frac{1}{\lambda_3-2\lambda_1-\overline{\lambda}_1}, \frac{1}{\overline{\lambda}_1+2{\lambda}_1-\overline{\lambda}_3}
\right),\\[0.3cm]
K_{(1,2)}:=&\frac{\sin(x)}{4(\overline{\lambda}_1-\lambda_1)}\left(\frac{9\kappa}{2\overline{\lambda}_1},0 
\right)\\
&+\frac{3\kappa\sin(3x)}{4(\overline{\lambda}_3-\lambda_3)}\left(
\frac{1}{\lambda_3-2\overline{\lambda}_1-\lambda_1}, \frac{1}{\lambda_1+2\overline{\lambda}_1-\overline{\lambda}_3}\right).
\end{split}
\end{equation}
Now, we may analyze the reduced dynamics, given by the equation
\begin{equation}\label{EqR}
\dot{z}=\lambda_1 z+R_0 z^2\overline{z},
\end{equation} 
or, writing the complex variable $z=x+\ri y$ and the the eigenvalue $\lambda_1=A+B\ri$, we obtain the equivalent system
\begin{equation}
\left(\begin{array}{c}
\dot{x}\\ \dot{y}
\end{array}\right)=\left(\begin{matrix}
A & -B\\ B & A
\end{matrix}\right) \left(\begin{array}{c}
x\\ y
\end{array}\right)+\frac{9\kappa}{8B}(x^2+y^2)\left(\begin{array}{c}
-y\\ x
\end{array}\right).
\end{equation}
\begin{figure}
	\centering
	\includegraphics[scale=0.7]{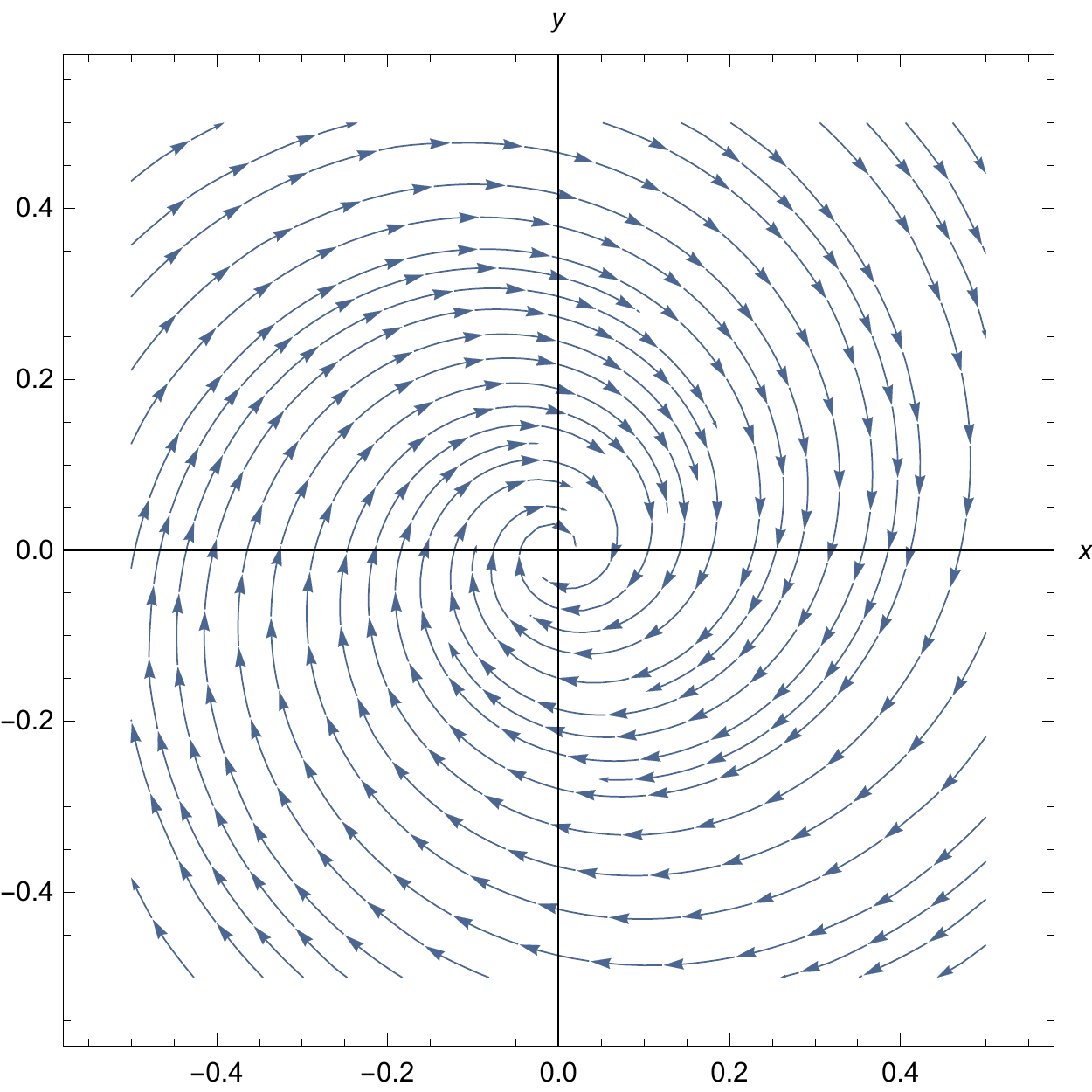}
	\caption{A typical phase portrait for the reduced system (\ref{eqpolar}) with parameter values $\alpha=1$, $\beta=0.6$, $\gamma=1$, $\delta=0.5$, $\mu=1$, $\kappa=1$, $A=-0.275$ and $B=-0.9614$.}
	\label{PlotReduction2}
\end{figure}
Rewriting system (\ref{EqR}) in polar coordinates $z(t)=r(t)e^{\ri\theta(t)}$, we obtain the equations
\begin{equation}\label{eqpolar}
\begin{cases}
\dot{r}&=Ar\\
\dot{\theta}&=B+\displaystyle\frac{9\kappa}{8B}r^2.
\end{cases}
\end{equation}
System (\ref{eqpolar}) can be integrated explicitly to
\begin{equation}
r(t)=e^{At}r_0,\quad \theta(t)=\theta_0+Bt+\frac{9\kappa r_0}{16BA}(e^{2At}-1),
\end{equation}
where $(r(0),\theta(0))=(r_0,\theta_0)$.\\
The right hand side of (\ref{eqpolar}) in the $\theta$- variable provides a measure for the instantaneous oscillation frequency of the full system at leading order. We therefore set
\begin{equation}
\Omega(r)=B+\displaystyle\frac{9\kappa}{8B}r^2.
\end{equation} 
We define the \textit{nominal instantaneous amplitude} as the quantity
\begin{equation}
\text{Amp}(r)=\sqrt{\frac{1}{2\pi}\int_0^{2\pi}|\mathbf{V}\cdot K(z_1(r,\theta),z_2(r,\theta))|^2\,d\theta},
\end{equation}
where $\mathbf{V}$ is the change of coordinates to the physical variables, defined in (\ref{defV}). In our example, $\text{Amp }(r)= 2r+\mathcal{O}(r^2)$. Now, the \textit{backbone curve} associated to the dynamics of the reduced system (\ref{eqpolar}) is defined as $\mathbf{B}:\mathbb{R}\to\mathbb{R}^2$,
\begin{equation}
\mathbf{B}(r)=\left(\begin{array}{c}
\Omega(r)\\ \text{Amp}(r)
\end{array}\right).
\end{equation}
Backbone curves for different parameter-values of $\kappa$ are depicted in Figure \ref{Backbone}.
\begin{figure}[h]
	\centering
	\includegraphics[scale=0.5]{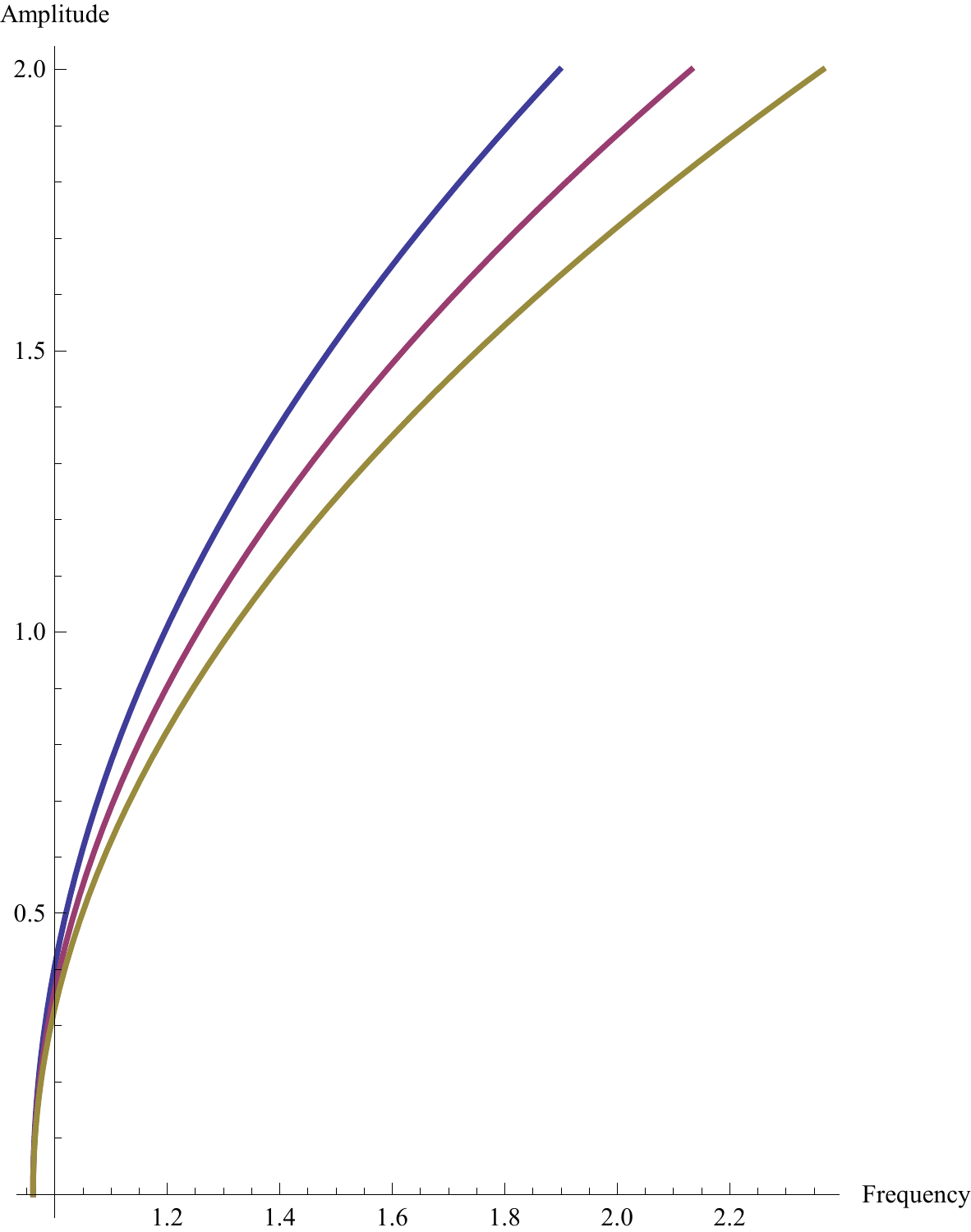}
	\caption{Backbone curves for the reduced system (\ref{eqpolar}) with parameter values $\alpha=1$, $\beta=0.6$, $\gamma=1$, $\delta=0.5$, $\mu=1$, $A=-0.275$, $B=-0.9614$ and $\kappa=0.8$ (blue), $\kappa=1$ (magenta) and $\kappa=1.2$ (green).}
	\label{Backbone}
\end{figure}

\subsection{An example with weak external time-periodic forcing}
To illustrate Theorem \ref{mainthmdifeps}, we consider an example with weak external time-periodic forcing. As in the previous example, we assume a cubic nonlinearity
\begin{equation}
f(u)=-\kappa u^3
\end{equation}
and set the internal damping to
\begin{equation}
\beta=\frac{4\delta\mu}{1-3\mu},
\end{equation}
while as time-periodic external forcing, we choose 
\begin{equation}\label{force}
h(x,t)=\cos(\omega t)\sin(x).
\end{equation}
In the eigenbasis of $A$, the forcing (\ref{force}) takes the form
\begin{equation}
\mathbf{h}(x,\theta)=\frac{1}{\overline{\lambda}_1-\lambda_1}\left(\begin{array}{c}
-1\\1
\end{array}\right)\cos(\theta)\sin(x).
\end{equation}
We assume that the forcing frequency $\omega$ is not in resonance with the eigenfrequencies $\lambda_n$ of the linear part of (\ref{mainequ}) in \ref{specA}, i.e., 
\begin{equation}\label{resomega}
\frac{\text{ Im }(\lambda_n)}{\omega}\notin\mathbb{Z},
\end{equation}
for all $n\in\mathbb{N}^{+}$.
In the following, we will compute the dynamics on the SSM up to first order in $\varepsilon$.\\
As the parameter $\varepsilon$ is small, the eigenvalue $\lambda_1$ in (\ref{l1}) either perturbs into an eigenvalue $\lambda_1(\varepsilon)$ with geometric multiplicity two or splits into two eigenvalues, $\lambda_{1,1}(\varepsilon)$ and $ \lambda_{1,2}(\varepsilon)$, both with geometric multiplicity one, respectively, cf. Appendix II, Lemma \ref{perteig}. In any case, we choose as our spectral subspace the two-dimensional eigenspace associated to the the perturbed eigenvalue or the split-eigenvalues. By analytical spectral perturbation theory, cf. Appendix II, Proposition \ref{perspace}, the perturbed spectral subspace is $\varepsilon$-close to the unperturbed spectral subspace defined in (\ref{specun}), i.e.
\begin{equation}\label{eper}
\mathcal{E}_\varepsilon=\mathcal{E}+\mathcal{O}(\varepsilon)\cong \mathbb{C}^2.
\end{equation}
Denote the coordinates in the space $\mathcal{E}_\varepsilon$ as $z^\varepsilon=(z_1^\varepsilon,z_2^\varepsilon)$ and denote the coordintes in the space $\mathcal{E}$ again as $z=(z_1,z_2)$. It follows from Pythagoras theorem and equation (\ref{eper}) that $$|z^\varepsilon|=\sqrt{|z|^2+\mathcal{O}(\varepsilon^2)}.$$
Hence, using $\sqrt{1+x}=1+\mathcal{O}(x)$, for $|x|<1$, we have that 
\begin{equation}\label{zeps}
z^\varepsilon=z+\mathcal{O}(\varepsilon^2).
\end{equation}
In accordance with Theorem \ref{mainthmdifeps} and equation (\ref{mainthmdifeqeps}), we assume that the unique spectral submanifold $\mathcal{W}(\mathcal{E})$ can be parametrized by an analytic function $K_\varepsilon:\mathcal{E}_\varepsilon\times S_\omega\to \mathcal{H}$ given by
\begin{equation}\label{defKeps}
K_\varepsilon(\theta,z^\varepsilon)=\sum_{|n|=1}^{\infty} K_n(\theta,\varepsilon)(z^\varepsilon)^n,
\end{equation}
for $n=(n_1,n_2)$ and the coefficients themselves depend analytically upon $\varepsilon$:
\begin{equation}\label{expK}
 K_n(\theta,\varepsilon)=\sum_{m=0}^{\infty}K_n^m(\theta)\varepsilon^m,
\end{equation}
for all $\theta\in S_\omega$. The dynamics on the spectral submanifold can be described by a polynomial
\begin{equation}\label{defReps}
R_\varepsilon=R^0+\varepsilon R^1 +\mathcal{O}(\varepsilon^2),
\end{equation}
where $R_0$ is given by the polynomial obtained in the unforced example, cf. (\ref{EqR}), i.e.,
\begin{equation}
R^0(z)=\left(\begin{array}{c}\lambda_1 z+\frac{9\kappa\ri}{8\text{ Im }\lambda_1} z^2\overline{z}\\[0.2cm] \overline{\lambda}_1 \overline{z}-\frac{9\kappa\ri}{8\text{ Im }\lambda_1} z\overline{z}^2
\end{array}\right).
\end{equation}
From equation (\ref{zeps}), we know that
\begin{equation}
\begin{split}
&K_\varepsilon=\left.K_\varepsilon\right|_{z^\varepsilon\mapsto z}+\mathcal{O}(\varepsilon^2),\\[0.3cm]
&DK_\varepsilon=\left.\frac{\partial K_\varepsilon}{\partial z^\varepsilon}\right|_{z^\varepsilon\mapsto z} + \mathcal{O}(\varepsilon^2),\\[0.3cm]
&D_\theta K_\varepsilon=\left.\frac{\partial K_\varepsilon}{\partial\theta}\right|_{z^\varepsilon\mapsto z}+ \mathcal{O}(\varepsilon^2).
\end{split}
\end{equation}
This implies that equation (\ref{mainthmdifeqeps}) can be written as
\begin{equation}\label{equK}
\begin{split}
&A\cdot K_\varepsilon(z)+G(K_\varepsilon(z))=\\
&=DK_\varepsilon(z)\cdot R^0+\varepsilon DK_\varepsilon(z)\cdot R^1+\omega D_\theta K_\varepsilon(z)-\varepsilon\mathbf{h}(x,\theta)+\mathcal{O}(\varepsilon^2),
\end{split}
\end{equation}
where, again, $z\in\mathcal{E}$. Since $K_\varepsilon=K_0+\mathcal{O}(\varepsilon)$, for $K_0$ as defined in equation (\ref{defK}), it follows from equation (\ref{equK}) and the first-order tangency condition (\ref{Kcoef}) that
\begin{equation}\label{K_1^0}
K_{(1,0)}^0=\left(\begin{array}{c}
1\\0
\end{array}\right)\sin(x),\quad K_{(0,1)}^0=\left(\begin{array}{c}
0\\1
\end{array}\right)\sin(x).
\end{equation}
In particular, the parametrization $K_\varepsilon$ does not depend upon $\theta$ at zeroth order in $\varepsilon$. Equation (\ref{equK}) at order $z^0$ becomes
\begin{equation}\label{z^0}
0=\varepsilon\left(\begin{array}{c}
K_{(1,0)}\\ K_{(0,1)}
\end{array}\right)\cdot R^1_0-\frac{1}{\overline{\lambda}_1-\lambda_1}\left(\begin{array}{c}
-1\\1
\end{array}\right)\cos(\theta)+\mathcal{O}(\varepsilon^2).
\end{equation}
Using equation (\ref{K_1^0}), we find that equation (\ref{z^0}) has the solution
\begin{equation}
R^1_0(\theta)=\frac{1}{\overline{\lambda}_1-\lambda_1}\left(\begin{array}{c}
-1\\1
\end{array}\right)\cos(\theta).
\end{equation}
Due to the persistence of the non-resonant nature of the perturbed eigenvalues, we may choose the dynamics on the spectral submanifold as
\begin{equation}\label{dynReps}
R_\varepsilon(z,\theta)=\left(\begin{array}{c}\lambda_1 z+\frac{9\kappa\ri}{8\text{ Im }\lambda_1} z^2\overline{z}\\[0.2cm] \overline{\lambda}_1 \overline{z}-\frac{9\kappa\ri}{8\text{ Im }\lambda_1} z\overline{z}^2
\end{array}\right)+\frac{1}{\overline{\lambda}_1-\lambda_1}\left(\begin{array}{c}
-1\\1
\end{array}\right)\cos(\theta)\varepsilon +\mathcal{O}(\varepsilon^2).
\end{equation}
We will now solve equation (\ref{equK}) for $K_\varepsilon$ up to order two in $z$. To exemplify the general computations, we only compute the expansion at order $z_1^2$. At order $z_2^2$ and $z_1z_2$, similar computations can be carried out.\\
At order zero in $\varepsilon$, equation (\ref{equK}) is solved by the unperturbed parametrization $K^0$. At order one in $\varepsilon$, equation (\ref{equK}) becomes
\begin{equation}\label{oeps1}
A\cdot K^1+\left.\frac{d}{d\varepsilon}\right|_{\varepsilon=0}(G(K_\varepsilon))=DK^1\cdot R^0+DK^0\cdot R^1+\omega D_\theta K^1-\mathbf{h}
\end{equation}
At order $z_1$ we find:
\begin{equation}\label{epsz1}
A\cdot K_{(1,0)}^1=\lambda_1K_{(1,0)}^1
+\omega \dot{K}_{(1,0)}^1.
\end{equation}
Due to the non-resonance condition (\ref{resomega}), the only periodic solution of (\ref{epsz1}) is given by 
\begin{equation}
K_{(1,0)}^1=0.
\end{equation}
Similarly, we find that 
\begin{equation}
K_{(0,1)}^1=0.
\end{equation}
At order $z_1^2$, equation (\ref{oeps1}) becomes
\begin{equation}
\begin{split}
A\cdot K^1_{(2,0)}=2\lambda_1K^1_{(2,0)}+\frac{\cos(\theta)}{\overline{\lambda}_1-\lambda_1}(K^0_{(2,1)}-3K^0_{(3,0)})+\omega\dot{K}^1_{(2,0)}.
\end{split}
\end{equation}
Using equation (\ref{Kcoef21}), and expanding
\begin{equation}
K^1_{n}=\sum_{l=1}^{\infty}\sum_{m\in\mathbb{Z}}K^{1,l,m}_ne^{\ri m\theta}\sin(lx),
\end{equation}
for $n=(n_1,n_2)$, we find that
\begin{equation}
\begin{split}
&\sum_{m\in\mathbb{Z}}^{\infty}\left(
\omega \ri m +\left(\begin{matrix}
\lambda_1 & 0\\0 & 2\lambda_1-\overline{\lambda}_1
\end{matrix}\right)\right)
K_{(2,0)}^{1,1,m}e^{\ri m \theta}=\frac{9\kappa\cos(\theta)}{4(\overline{\lambda}_1-\lambda_1)}\left(\begin{array}{c}
\frac{1}{2\lambda_1}\\ \frac{\overline{\lambda}_1-\lambda_1}{2\lambda_1(\overline{\lambda}_1-3\lambda_1)}\end{array}
\right),\\[0.3cm]
&\sum_{m\in\mathbb{Z}}^{\infty}\left(
\omega \ri m +\left(\begin{matrix}
2\lambda_1 -\lambda_3& 0\\0 & 2\lambda_1-\overline{\lambda}_3
\end{matrix}\right)\right)
K_{(2,0)}^{1,3,m}e^{\ri m \theta}=-\frac{3\kappa\cos(\theta)}{4(\overline{\lambda}_3-\lambda_3)}\left(\begin{array}{c}
\frac{\overline{\lambda}_1-\lambda_1}{(\lambda_3-2\lambda_1-\overline{\lambda}_1)(\lambda_3-3\lambda_1)}\\ \frac{\lambda_1-\overline{\lambda}_1}{(\overline{\lambda}_1+2{\lambda}_1-\overline{\lambda}_3)(3\lambda_1-\overline{\lambda}_3)}\end{array}
\right),
\end{split}
\end{equation}
which can be solved as
\begin{equation}
\begin{split}
K_{(2,0)}^1(\theta)=&\frac{9\kappa}{8(\overline{\lambda}_1-\lambda_1)}\left(\begin{array}{c}
\frac{1}{2\lambda_1(\lambda_1+\ri \omega)}\\ \frac{\overline{\lambda}_1-\lambda_1}{2\lambda_1(\overline{\lambda}_1-3\lambda_1)(2\lambda_1-\overline{\lambda}_3+\ri\omega )}\end{array}
\right)e^{\ri \theta}\sin(x)\\
&+\frac{9\kappa}{8(\overline{\lambda}_1-\lambda_1)}\left(\begin{array}{c}
\frac{1}{2\lambda_1(\lambda_1-\ri \omega)}\\ \frac{\overline{\lambda}_1-\lambda_1}{2\lambda_1(\overline{\lambda}_1-3\lambda_1)(2\lambda_1-\overline{\lambda}_3-\ri\omega )}\end{array}
\right)e^{-\ri \theta}\sin(x)\\
&+\frac{3\kappa}{8(\lambda_3-\overline{\lambda}_3)}
\left(\begin{array}{c}
\frac{\overline{\lambda}_1-\lambda_1}{(\lambda_3-2\lambda_1-\overline{\lambda}_1)(\lambda_3-3\lambda_1)(2\lambda_1-\lambda_3+\ri\omega)}\\ \frac{\lambda_1-\overline{\lambda}_1}{(\overline{\lambda}_1+2{\lambda}_1-\overline{\lambda}_3)(3\lambda_1-\overline{\lambda}_3)(2\lambda_1-\overline{\lambda}_3+\ri\omega)}\end{array}
\right)e^{\ri\theta}\sin(3x)\\
&+\frac{3\kappa}{8(\lambda_3-\overline{\lambda}_3)}
\left(\begin{array}{c}
\frac{\overline{\lambda}_1-\lambda_1}{(\lambda_3-2\lambda_1-\overline{\lambda}_1)(\lambda_3-3\lambda_1)(2\lambda_1-\lambda_3-\ri\omega)}\\ \frac{\lambda_1-\overline{\lambda}_1}{(\overline{\lambda}_1+2{\lambda}_1-\overline{\lambda}_3)(3\lambda_1-\overline{\lambda}_3)(2\lambda_1-\overline{\lambda}_3-\ri\omega)}\end{array}
\right) e^{-\ri\theta}\sin(3x).
\end{split}
\end{equation}
As in the unperturbed example, we can rewrite the dynamics (\ref{dynReps}) in polar coordinates
$z(t)=r(t)e^{\ri\phi(t)}$:
\begin{equation}\label{dynRepspolar}
\begin{cases}
\dot{r}&=Ar-\frac{\varepsilon}{B}\cos(\omega t)\sin(\phi)\\
\dot{\phi}&=B+\displaystyle\frac{9\kappa}{8B}r^2-\frac{\varepsilon}{Br}\cos(\omega t)\cos(\phi),
\end{cases}
\end{equation}
where again $\lambda_1=A+B\ri$.\\
We note that the overall dynamics of system  (\ref{dynRepspolar}) are close to the dynamics of the unperturbed system (\ref{eqpolar}), as it is also depicted in Figure \ref{Perturb}.
\begin{figure}[h]
	\centering
	\includegraphics[scale=0.25,trim={10cm 0 0 0}]{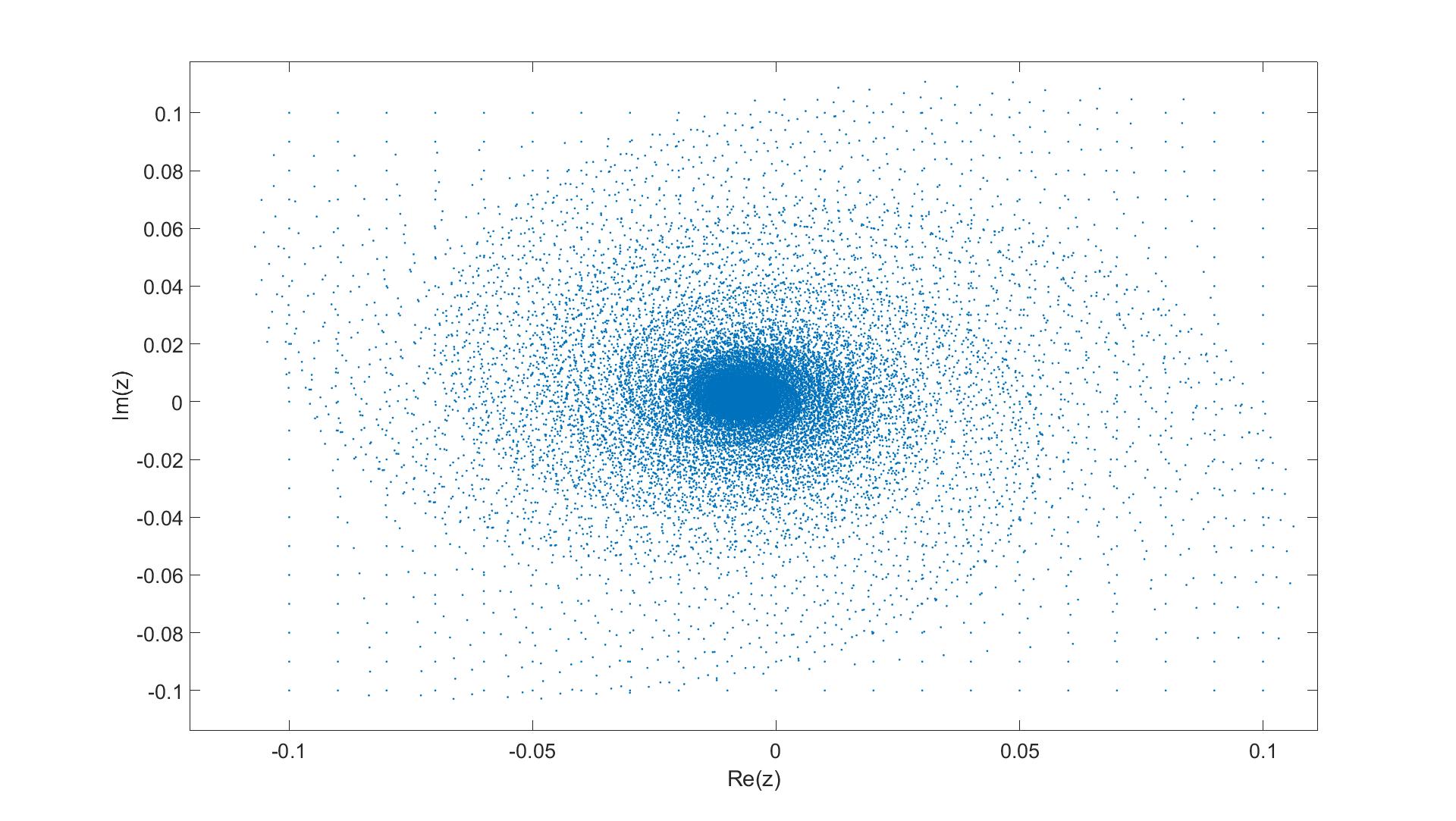}
	\caption{Sampled dynamics of system (\ref{dynRepspolar}) for the sampling time $T=2\pi$.}
	\label{Perturb}
\end{figure}

\subsection{Conclusion}
We have constructed a rigorous reduction of a nonlinear, damped-forced continuum beam model to a two-dimensional spectral submanifold (SSM). This SSM acts as a slow manifold both in the unforced and the time-periodically forced case, even though the underlying beam problem admits no clearly defined time scale differences embodied by small parameters. We have also pointed out why simpler beam models, such as Euler-Bernoulli beam theory, and other damping mechanisms, such as pure viscous damping or visco-elastic damping, do not qualify for the analysis given here.\\

For the Rayleigh beam analyzed here, we combined existence and uniqueness results, a careful analysis of the linearized spectrum and infinite-dimensional Poincaré map techniques with an abstract theorem for maps on Banach spaces by Cabr\'{e}, Fontich and de la Llave \cite{Cab2003}. We believe that our result is the first mathematically rigorous example of reducing an infinite-dimensional structural vibration problem exactly to a low-dimensional model. The analysis presented here also appears to be the first infinite-dimensional application of the parametrization method developed in \cite{Cab2003} for constructing invariant manifolds.\\

The analysis here justifies model reductions carried out on discretizations and Galerkin projections of the underlying PDE, but only with carefully chosen damping models. While viscous and viscoelastic damping enable model reduction without technical difficulties in finite-dimensional systems, they lead to either conceptual difficulties (lack of determinism for a non-stochastic linear vibration under viscoelastic damping) or technical difficulties (lack of distinguished modes for model reduction under viscous damping) at the level of the PDE.\\

The methods presented in this paper should be applicable to more general nonlinear PDEs with a time-reversible flow, such as non-homogeneous beam equations or nonlinear wave equations. It would be feasible to extend the analysis to PDEs that only admit a semi-group as a flow map, such as the heat equation or reaction-diffusion systems. This, however, would also require new abstract invariant manifold results that do not depend on the invertibility of the linearization.\\

\section{Appendix I: Well-posedness and global existence}
In this section, we give the technical background for the evolution equation (\ref{mainequ}) in detail. We will prove well-posedness, even differential dependence on initial data, for system (\ref{mainequ}) and show that the solutions exist globally in time, provided that the time-dependent forcing term $\varepsilon h(t,x)$ is sufficiently well-behaved. We first consider the case of no external forcing ($\varepsilon=0$), then turn to the case $\varepsilon>0$. In our presentation of these classical applications of semigroup theory, we will be following \cite{engel1999one} and \cite{pazy1992semigroups}. \\
We recall that a pair of functions $(u,v)$ in some appropriately chosen space is called \textit{classical solution} to (\ref{mainequ}) if the map $t\mapsto (u(t),v(t))$ is continuously differentiable and satisfies equation (\ref{mainequV}) pointwise. We will refer to the flow map for the linear part of the right-hand side of (\ref{mainequV}) as
\begin{equation}\label{linflow}
(u_0,v_0)\mapsto e^{A t}(u_0,v_0).
\end{equation} That is to say, (\ref{linflow}) solves the initial value problem
\begin{equation}\label{linpro}
\begin{cases}
U_t=AU\\U(0)=U_0,
\end{cases}
\end{equation}
for $U=(u,v)$ and $U_0=(u_0,v_0)$. In section \ref{linequ}, we will show that problem (\ref{linpro}) indeed has a solution that depends continuously on the initial data. Once this is established, we may proceed with the analysis of the full system (\ref{mainequV}).\\
Since solving an evolution equation in the classical sense may be too restrictive, we introduce a weaker form of solution. This also permits us to use archetypal arguments from nonlinear analysis on Banach spaces, in particular fixed point arguments. A function $t\mapsto U(t):=(u(t),v(t))$ is called a \textit{mild solution} if it satisfies equation (\ref{mainequV}) in the integral sense
\begin{equation}\label{mild}
U(t)=e^{tA}U_0+\int_0^te^{(t-s)A}G_{\varepsilon}(u(s),x)\,ds,\quad t\geq 0,
\end{equation} 
where $U_0=(u_0,v_0)$; the one parameter family of operators $t\mapsto e^{tA}$ again denotes the solution to the linearized equation (\ref{linpro}); and $G_{\varepsilon}(u,x)=(0,f(u)+\varepsilon h(x,t))$. A priori, it is not clear that equation (\ref{mainequV}) possesses a solution of any kind at all. We will show that a solution actually exists and depends in a differentiable fashion on the initial conditions due to the properties of the nonlinearity (\ref{AsA}).\\

\subsection{The linearized equation}\label{linequ}

In order to show that a mild solution to equation (\ref{mainequV}) exists, we first prove that the matrix of operators (\ref{defMA}) generates a $C^0$-semigroup of contractions. We recall that a one-parameter family of operators $T:\mathbb{R}^+\to \mathcal{L}(X)$ on some Banach space $X$ is called a \textit{strongly continuous semigroup}, or for short $C^0$-\textit{semigroup}, if
\begin{itemize}
	\item $T(0)=Id_X$.
	\item $T(s+t)=T(s)T(t)$, for all $s,t\geq 0$.
	\item For all $x\in X : \|T(t)x-x\|\to 0$ as $t\to 0$.
\end{itemize}
A $C^0$-semigroup $\{T(t)\}_{t\in\mathbb{R}^+}$ is called \textit{contractive} if in addition $$\|T(t)\|\leq 1$$
for all $t\geq 0$. Here, we have denoted the set of all bounded linear operators on $X$ by $\mathcal{L}(X)$. A linear operator generates a $C^0$-semigroup if and only if the underlying evolution equation is well-posed, cf. \cite{engel1999one} and \cite{pazy1992semigroups}.\\
For the forthcoming analysis of equation (\ref{mainequV}), we choose the Hilbert space 
\begin{equation}\label{defMH}
\mathcal{H}:=H^1_0(0,\pi)\times L^2(0,\pi).
\end{equation}
We introduce the following inner product on the space $H^1_0(0,L)$, depending on three parameters $\alpha, \gamma,\mu\geq0$ :
\begin{equation}\label{inprodalpha}
\langle f,g\rangle_{\alpha,\gamma,\mu}:=\sum_{n=1}^{\infty}\left(\frac{\alpha n^4+\gamma}{1+\mu n^2}\right)\hat{u}_n\hat{v}_n^{*}.
\end{equation}
For $\alpha, \gamma, \mu>0$, the norm induced by the inner product (\ref{inprodalpha}) is equivalent to the standard norm (\ref{defHs}) on $H^1(0,\pi)$, as can be seen by direct comparison. 
Next, for fixed $\alpha, \gamma, \mu>0$, we endow the space $\mathcal{H}$ with the inner product
\begin{equation}\label{defIMH}
\langle U_1,U_2\rangle_{\mathcal{H}}:=\langle u_1,u_2\rangle_{\alpha,\gamma,\mu}+\langle v_1,v_2\rangle_{L^2},
\end{equation}
where $U_1=(u_1,v_1), U_2=(u_2,v_2)\in\mathcal{H}$.\\
The domain of definition of the matrix of operators (\ref{defMA}) is the dense subspace $H^2(0,\pi)\times L^2(0,\pi)\cap H^1_0(0,\pi)\times L^2_0(0,\pi)\subset\mathcal{H}$, denoted by $\mathcal{D}(A)$.
\begin{theorem}\label{semi}
	The matrix of operators $A$ generates a $C^0$-semigroup of contractions.
\end{theorem}	
\begin{proof}
Let us show that $A$ is dissipative for $U=(u,v)\in\mathcal{D}(A)$. First, we expand the functions $u,v$ in Fourier-Sine series as
$$u(x)=\sum_{n=1}^{\infty}\hat{u}_n \sin(nx),\quad v(x)=\sum_{n=1}^{\infty}\hat{v}_n \sin(nx).$$ Then, using the inner product on $\mathcal{H}$ defined in (\ref{defIMH}), Parseval's formula (\ref{ParsI}) and the characterization of dissipativity in (\ref{dissHi}) we can write
\begin{equation*}
\begin{split}
\langle A U,U\rangle_{\mathcal{H}}&=\langle u,v\rangle_{\alpha,\gamma,\mu}+\langle v,\mathcal{M}^{-1}(-\alpha u_{xxxx}-\gamma u+\beta v_{xx}-\delta v)\rangle_{L^2}\\
&= \sum_{n=1}^{\infty}\left[\frac{(\alpha n^4+\gamma)\hat{u}_n\hat{v}_n^*-(\alpha n^4+\gamma)\hat{u}_n\hat{v}_n^*-(\beta n^2+\delta)|\hat{v}_n|^2}{1+\mu n^2}\right]\\
&=-\sum_{n=1}^{\infty}\frac{(n^2\beta+1)|\hat{v}_n|^2}{1+\mu n^2}\leq 0.
\end{split}
\end{equation*}
This proves dissipativity for the operator $A$.\\
Next, we calculate the spectrum of $A$ in $\mathcal{H}$. The operator $A-\lambda I$ is not invertible if and only if for some $n\in\mathbb{N}^{+}$ we have
\begin{equation}
\det\left(\begin{matrix}
-\lambda & 1\\ -\frac{\alpha n^4+\gamma}{1+\mu n^2} & -\frac{(\beta n^2+\delta)}{1+\mu n^2}-\lambda\end{matrix}\right)=0.
\end{equation}
Computing the zeros of the corresponding characteristic polynomial, we find that
\begin{equation}\label{eigval}
\lambda^{\pm}_n= -\frac{\beta n^2+\delta}{2+2\mu n^2}\pm \sqrt{\left(\frac{\beta n^2+\delta}{2+2\mu n^2}\right)^2-\frac{\alpha n^4+\gamma}{1+\mu n^2}},
\end{equation}
and hence $\sigma(A)=\left\{\lambda^{\pm}_n\right\}_{n\in\mathbb{N}^{+}}.$ Since the spectrum is a countable set of isolated eigenvalues, $\lambda_0 I-A$ is surjective whenever $\lambda_0\notin\sigma(A)$. Applying the Lummer-Phillips theorem, we conclude that the operator $A$ generates a $C^0$-semigroup of contractions, cf. \cite{engel1999one}.
\end{proof}
Expanding the initial conditions as Fourier series
$$u_0(x)=\sum_{n=1}^{\infty}\hat{u}_n^0 \sin(nx),\quad v_0(x)=\sum_{n=1}^{\infty}\hat{v}_n^0 \sin(nx),$$ and assuming that $\lambda_n^{+}\neq \pm\lambda_n^{-}$ for all $n\in\mathbb{N}^{+}$, we may write the semi-flow map generated by $A$ explicitly as
\begin{equation}\label{explicit}
\begin{split}
e^{tA} \left(\begin{array}{c}u_0\\v_0
\end{array}\right)=&\left(\sum_{n=1}^{\infty}\frac{(\lambda_n^-e^{\lambda_n^{+}t}-\lambda_n^+e^{\lambda_n^-t})\hat{u}_n^0+(e^{\lambda_n^-t}-e^{\lambda_n^{+}t})\hat{v}_n^0}{\lambda_n^--\lambda_n^+}\sin(nx),\right.\\[0.33cm]
&\left.
\sum_{n=1}^{\infty}\frac{\lambda_n^+\lambda_n^-(e^{\lambda_n^+t}-e^{\lambda_n^-t})\hat{u}_n^0+(\lambda_n^-e^{\lambda_n^-t}-\lambda_n^+e^{\lambda_n^{+}t})\hat{v}_n^0}{\lambda_n^--\lambda_n^+}\sin(nx)\right).
\end{split}
\end{equation}

\subsection{The full equation}
Now that we know that the matrix of operators (\ref{defMA}) generates a $C^0$-semigroup, we can proceed with the analysis of the full nonlinear and non-autonomous system (\ref{mainequV}).
First, we use the following theorem to show that the nonlinear equation (\ref{mainequV}) is well posed and the flow depends on the initial data in a differentiable fashion. 

\begin{theorem}\label{ExUn}
Let $A$ generate a $C^0$-semigroup on $\mathcal{H}$. If  $G:[0,T]\times\mathcal{H}\to\mathcal{H}$ continuous in $t$ on the interval $[0,T]$ and uniformly Lipschitz continuous on $\mathcal{H}$. Then for any $U_0\in\mathcal{D}(A)$, the initial value problem (\ref{mainequV}) has a unique mild solution.\\
If the forcing $G$ is even continuously differentiable from $[0,T]\times X$ into $X$, then the mild solution with $U_0\in\mathcal{D}(A)$ is also a classical solution.
\end{theorem}
A proof based on fixed-point arguments can be found in \cite{pazy1992semigroups}.\\
 We now have to show that the function $f:\mathcal{D}(A)\to\mathcal{H}$ is continuously differentiable as a map on Hilbert spaces. To this end, we will need the following result.
\begin{lemma}\label{Sobf}
Let $f:\mathbb{R}\to\mathbb{R}$ satisfy the assumptions in \ref{AsA} and let $p\geq m$. Then the nonlinear operator
\begin{equation*}
\begin{split}
f: &L^p(0,\pi)\to L^{\frac{p}{m}}(0,\pi),\\
 &u\mapsto f(u),
\end{split}
\end{equation*}
between Hilbert spaces is well-defined and $r$-times continuously differentiable in the Fr\'{e}chet sense, with derivative
\begin{equation}
f'(u)\cdot v=f'(u)v
\end{equation}
for all $v\in L^p(0,\pi)$.
\end{lemma}
A proof (under weaker assumptions) can be found in  \cite{cazenave1998introduction}. Now we have to relate the $L^p$-space in the previous lemma to the Sobolev space $H^s$. This is done in the following
\begin{lemma}\label{LiHs}
Let $u\in H^1(0,\pi)$ Then the inequality $$\|u\|_{L^\infty(0,\pi)}\leq C\|u\|_{H^1(0,\pi)}$$
holds true for some $C>0$.
\end{lemma}
\begin{proof}
Expanding $u$ as a Fourier-sine series, for any $x\in [0,\pi]$, we have the estimate  $$|u(x)|\leq\sum_{n=1}^{\infty}|\hat{u}_n|\leq\left(\sum_{n=1}^{\infty}\frac{1}{1+\tilde{n}^2}\right)^{\frac{1}{2}}\left(\sum_{n=1}^{\infty}(1+\tilde{n}^2)|\hat{u}_n|^2\right)^{\frac{1}{2}}=C\|u\|_{H^1(0,\pi)},$$
where we have used the Cauchy-Schwartz inequality and the convergence of the series $\sum_{n=1}^{\infty}\frac{1}{1+\tilde{n}^2}$. Taking the supremum of the left-hand side of this inequality then proves the claim.
\end{proof}
We can now deduce that the map $f:H^s(0,\pi)\to L^2(0,\pi), u\mapsto f(u)$, is differentiable for all $s\geq 1$. Indeed, if we set $p=2m$ in Lemma \ref{Sobf} and recall that $L^\infty(0,\pi)\subset L^p(0,\pi)$ for all $p\geq 1$, it follows from Lemma \ref{LiHs}, that $H^s(0,\pi)\subset L^{2m}(0,\pi)$ for all $s\geq 1$ and all $m>1$. By Theorem \ref{ExUn}, the initial value problem (\ref{mainequV}) has a unique, classical solution.\\
Theorem \ref{ExUn} guarantees existence and uniqueness of a local solution to equation (\ref{mainequV}). We now show that a solution to equation (\ref{mainequV}) exists for all times $t\geq 0$, provided that the external forcing is sufficiently well-behaved.
\begin{proposition}\label{global}
	Assume that the external forcing in (\ref{mainequ}) satisfies Assumption \ref{Asforce}.
	Then any solution to equation (\ref{mainequ}) exists for all times.
\end{proposition}
\begin{proof}
	The proof relies upon a straight forward energy estimate for the quantity
	\begin{equation}\label{energy}
	\mathcal{E}(u)=\frac{1}{2}\int_0^\pi[ u_t^2+\alpha u_{xx}^2+\gamma u^2+\mu u_{tx}^2-2F(u)]\, dx,
	\end{equation}
where we have set $$F(x):=\int_0^xf(\xi)\,d\xi.$$
Indeed, by equation (\ref{HsInt}) and by the definition of the $\alpha$-$\gamma$-$\mu$-norm in (\ref{inprodalpha}), we find that 
\begin{equation}\label{energyest}
\begin{split}
\frac{1}{2}\|(u(t),v(t))\|_{\mathcal {H}}^2 &=\frac{1}{2}\sum_{n=1}^{\infty}\left(\frac{\alpha n^4+\gamma}{1+\mu n^2}\right)|\hat{u}_n|^2+|\hat{v}_n|^2\\
&\leq\frac{1}{2}\sum_{n=1}^{\infty}\left(\alpha n^4+\gamma\right)|\hat{u}_n|^2+|\hat{v}_n|^2\\
&\leq \frac{1}{2}\sum_{n=1}^{\infty}\left(\alpha n^4+\gamma\right)|\hat{u}_n|^2+(1+\mu n^2)|\hat{v}_n|^2
-\int_0^{\pi}F(u(x))\, dx\\
&=\mathcal{E}(u)
\end{split}
\end{equation}
by assumption (\ref{below}).\\
We will now show that the energy (\ref{energy}) is decreasing in time. Using equation (\ref{mainequ}) and integration by parts to shift the derivatives, we obtain the estimate
\begin{equation}\label{dissenergy}
\begin{split}
	\frac{d}{dt}\mathcal{E}(u)&=\int_0^\pi u_t(u_{tt}-\mu u_{ttxx}+\alpha u_{xxxx}+\gamma u-f(u))\, dx\\
	&=\int_0^\pi u_t(\beta u_{txx}-\delta u_t+\epsilon h(x,t))\, dx\\
	&\leq -\beta\|u_{tx}\|_{L^2}^2-\delta\|u_{t}\|_{L^2}^2+\varepsilon\|u_t\|_{L^2}^2\|h\|_{L^2}^2\\
	&\leq \|u_t\|_{L^2}^2(\varepsilon H_0-\delta)\\
	&\leq 0,
\end{split}
\end{equation}
where we have used the Cauchy-Schwartz inequality in the third line and Assumption \ref{Asforce} in the forth line. The estimate in the last line holds true for sufficiently small $\varepsilon$. Since $\mathcal{E}$ decreases in time, it follows from (\ref{energyest}), that the norm of a solution $(u,v)$ stays bounded for all times. This proves global existence.
\end{proof}

\begin{remark}
	The estimate (\ref{dissenergy}) actually shows that the mechanical system (\ref{mainequ}) loses energy in time - this is due to the damping term $\delta u_t$. The dissipation of energy outperforms the periodic forcing if the mean kinetic energy is sufficiently small.
\end{remark}
\begin{remark}
For later computations, we note that if $\|U_0\|_{\mathcal{H}}\lesssim \varepsilon$, then $\|U(t)\|\lesssim\varepsilon$ for all $t\geq 0$. Indeed, by (\ref{energyest}) and (\ref{dissenergy}), we have that 
\begin{equation}\label{estimate}
\begin{split}
\|U(t)\|_{\mathcal{H}}&\leq \|U_0\|_{\mathcal{H}}-\int_0^LF(u_0)\, dx\\
&\lesssim \varepsilon +\int_0^\pi|u_0|^{m+1}\, dx\\
&\lesssim \varepsilon +\|u_0\|_{L^2(0,\pi)}\lesssim \varepsilon,
\end{split}
\end{equation}
where we have used the fact that $\varepsilon$ is small and that $L^q(0,\pi)\subseteq L^p(0,\pi)$ if $q\geq p$.
\end{remark}

\subsection{Fixed points of the Poincar\'e map}
Because of the presence of an $\omega$-periodic forcing term, we introduce the phase $\theta\in S^1$ and consider the equivalent autonomous dynamical system
\begin{equation}\label{suspend}
\begin{cases}
U_t=AU+\mathbf{f}(U)+\varepsilon \mathbf{h}(x,\theta),\\
\theta_t=\omega,
\end{cases}
\end{equation}
with $\mathbf{f}(u,v)=(0,f(u))$ and $\mathbf{h}(x,t)=(0,h(x,t))$. This system is equivalent to our equation (\ref{mainequ}), but a fixed point of the returned map of equation (\ref{suspend}) now corresponds to a periodic orbit of (\ref{suspend}). Due to the presence of time-periodic forcing, the trivial solution $U(x,t)=(0,0)$ is no longer a fixed point of the flow map. However, we can prove that the suspended system (\ref{suspend}) admits a fixed point for its Poincar\'e map, as long as $\varepsilon$ is small.
Let $F_{\varepsilon}^{\frac{2\pi}{\omega}}:\mathcal{H}\times S_{\omega}\to\mathcal{H}\times S_{\omega}$ be the flow map of system (\ref{suspend}). For any fixed $\omega_0\in S_\omega$, we define the \textit{Poincar\'e map} with base $\omega_0$ as
\begin{equation}\label{Poin}
P_{\varepsilon}:\mathcal{H}\to\mathcal{H},\quad U_0\mapsto \pi_{\mathcal{H}}(F^{\omega}_{\varepsilon}(U_0,\omega_0)),
\end{equation}
where $\pi_{\mathcal{H}}:\mathcal{H}\times S_{\omega}\to\mathcal{H}$ is the projection on our underlying Hilbert space. The implicit dependence of $P_\varepsilon$ upon $\omega_0$ is suppressed for notational reasons in the following. \\
The assumptions made in (\ref{AsA}) and (\ref{Asforce}) are strong enough to guarantee well-posedness and global existence for the suspended system (\ref{suspend}), as we may deduce from the well-posedness and global existence result for equation (\ref{mainequV}).\\
Next, we observe that the linearization of equation (\ref{mainequ}) with forcing present,
\begin{equation}\label{linmain}
w_{tt}-\mu w_{ttxx}=-\alpha w_{xxxx}+\beta w_{txx}-\gamma w-\delta w_t+\varepsilon h,
\end{equation}
admits a time-periodic solution of period $\omega$. To see this, we seek a time-periodic solution and expand $w$ as well as $h$ as Fourier series, both in the $x$- and the $t$-variable, as $$w(x,t)=\sum_{m\in\mathbb{Z}}\sum_{n=1}^{\infty}\hat{w}_{n,m}e^{\frac{2\pi \ri m}{\omega}t}\sin(nx),\quad h(x,t)=\sum_{m\in\mathbb{Z}}\sum_{n=1}^{\infty}\hat{h}_{n,m}e^{\frac{2\pi \ri m}{\omega}t}\sin(nx),$$
and insert these expressions into equation (\ref{linmain}). We find that in order to obtain a time-periodic solution, we should set
\begin{equation}\label{Fourlin}
\hat{w}_{n,m}=\frac{\varepsilon\hat{h}_{n,m}}{\ri\tilde{m}(\delta+\beta n^2)+\alpha n^4+\gamma-\tilde{m}^2(\mu{n}^2+1)},
\end{equation}
provided the denominator is non-zero for any $m\in\mathbb{Z}$ and $n\in\mathbb{N}^{+}$. Here we have set $\tilde{m}:=\frac{2\pi m}{\omega}$. Equivalently, passing to a vector formulation of (\ref{linmain}) with $W=(w,w_t)$, we know that the system
\begin{equation}\label{linmainV}
W_t=AW+\varepsilon\mathbf{h}
\end{equation}
possesses a time-periodic solution of period $\omega$.\\
In order to prove that the Poincar\'e map possesses a fixed point, we impose the following non-resonance condition relating the spectrum of the linear system (\ref{linmain}) to the period of the external forcing
\begin{assumption}
	The spectrum of the linear flow at time $\omega$ does not contain $1$, i.e., 
	\begin{equation}\label{nonres}
	1\notin \sigma(e^{\omega A}).
	\end{equation}
\end{assumption}

We may write assumption (\ref{nonres}) equivalently as
\begin{equation}
\frac{2\pi\ri l}{\omega}\neq  -\frac{\beta n^2+\delta}{2+2\mu n^2}\pm \sqrt{\left(\frac{\beta n^2+\delta}{2+2\mu n^2}\right)^2-\frac{\alpha n^4+\gamma}{1+\mu n^2}}
\end{equation}
for any $n,l\in\mathbb{N}$. Note that for $\beta, \delta >0$, the above assumption is always satisfied.\\
Using a perturbative argument in the proof of the existence of a fixed point for $P_\varepsilon$, we have to show that the derivative of the flow map of $(\ref{mainequV})$ for $\varepsilon$ small is close to the derivative of the flow map at $\varepsilon=0$. This will be achieved in the following
\begin{lemma}\label{derclose}
	Let $F_{\varepsilon}^t$ be the flow map of equation (\ref{mainequV}) at time $t>0$. If we assume that the nonlinearity $f$ satisfies Assumption \ref{AsA}, while the forcing satisfies Assumption $\ref{Asforce}$, then, for any fixed time $t>0$, the estimate
	\begin{equation}\label{estder}
\|(DF^t_{0}(U)-DF^t_{\varepsilon}(U))\cdot V\|_{\mathcal{H}}
\lesssim \|V\|_{\mathcal{H}},
	\end{equation}
	holds for any $U,V\in\mathcal{H}$, provided $\varepsilon>0$ is sufficiently small. 
\end{lemma}
\begin{proof}
To show the estimate (\ref{estder}), we first note that the derivative $DF^t_{\varepsilon}$ satisfies the integral equation
\begin{equation}\label{DF}
DF_{\varepsilon}^t(U)\cdot V=e^{tA}\cdot V+\int_0^te^{(t-s)A}\nabla_{U}G_{\varepsilon}(F^t_{\varepsilon}(U))DF_{\varepsilon}^t(U)\cdot V\, ds,
\end{equation}
as it can be seen by taking the Fr\'echet derivate of equation (\ref{mild}) and using the regularity properties of $f$ derived in Lemma \ref{Sobf}. We rewrite (\ref{DF}) as
\begin{equation*}
\begin{split}
&(DF^t_{\varepsilon}(U)-DF^t_{0}(U))\cdot V \\
&= \int_0^t  e^{(t-s)A}\Big(\nabla_{U}G_{\varepsilon}(F^t_{\varepsilon}(U))DF_{\varepsilon}^t(U)-\nabla_{U}G_{0}(F^t_{0}(U))DF_{0}^t(U)\Big)\cdot V\, ds\\
&=\int_0^te^{(t-s)A}\Big(\nabla_{U}G_{\varepsilon}(F^t_{\varepsilon}(U))\Big(DF_{\varepsilon}^t(U)-DF_{0}^t(U)\Big)\Big)\cdot V\, ds\\
&+\int_0^te^{(t-s)A}\Big(\Big(\nabla_{U}G_{\varepsilon}(F^t_{\varepsilon}(U))-\nabla_{U}G_{0}(F^t_{0}(U))\Big)DF_{0}^t(U)\Big)\cdot V\, ds.
\end{split}
\end{equation*}
Using the fact that $\|DF_{0}^t(U)\cdot V\|_{\mathcal{H}}=\|e^{tA}\cdot U\|_{\mathcal{H}}\lesssim \|U\|_{\mathcal{H}}$ for any $U\in\mathcal{H}$ along with Assumption \ref{AsA}, we obtain
\begin{equation}\label{estgron}
\|(DF^t_{\varepsilon}(U)-DF^t_{0}(U))\cdot V\|_{\mathcal{H}}\lesssim \int_0^t \|(DF^t_{\varepsilon}(U)-DF^t_{0}(U))\cdot V\|_{\mathcal{H}}\, ds + \|V\|_{\mathcal{H}}.
\end{equation}
Applying Gronwall's inequality to equation (\ref{estgron}) proves the claim.
\end{proof}
We are now ready to prove the following
\begin{lemma}\label{fixed}
	For $\varepsilon>0$ small enough, the Poincar\'e map (\ref{Poin}) admits a unique fixed point.
\end{lemma}
We follow closely the argument provided in \cite{Holmes1981}, under slightly weaker assumptions.
\begin{proof}
Let $U_{per}(x,t;\varepsilon)$ denote the unique time-periodic solution to equation (\ref{linmainV}). Passing to the weak formulation, we know that $U^{per}$ satisfies the integral equation
\begin{equation}\label{Uper}
U_{per}(x,t;\varepsilon)=e^{tA}U_{per}(x,0;\varepsilon)+\varepsilon\int_0^te^{(t-s)A}\mathbf{h}(x,t)\, ds,
\end{equation}
together with the condition $U_{per}(x,\frac{2\pi}{\omega};\varepsilon)=U_{per}(x,0;\varepsilon)$ for any $x\in (0,\pi)$. We are looking for a solution to the integral equation
\begin{equation}\label{Ufix}
U(x,t;\varepsilon)=e^{tA}U(x,0;\varepsilon)+\varepsilon\int_0^te^{(t-s)A}\mathbf{h}(x,t)\, ds+\int_0^te^{(t-s)A}\mathbf{f}(U)\,ds,
\end{equation}
together with the condition that $U(x,\frac{2\pi}{\omega};\varepsilon)=U(x,0;\varepsilon)$ for all $x\in (0,\pi)$.  Let $B_{\varepsilon}$ denote the ball of radius $\varepsilon$ around the solution to the linear problem $U_{per}(x,0;\varepsilon)$ in the space $\mathcal{H}$. Subtracting equation (\ref{Uper}) from equation (\ref{Ufix}), we obtain
\begin{equation}\label{Uperfix}
U(x,t;\varepsilon)-U_{per}(x,t;\varepsilon)=e^{tA}\big(U(x,0;\varepsilon)-U_{per}(x,0;\varepsilon)\big)+\int_0^te^{(t-s)A}\mathbf{f}(U)\,ds,
\end{equation}
as we see by equation (\ref{Uperfix}).  We note that $U$ is a fixed point for the Poincar\'e map if and only if it is a fixed point for the functional
\begin{equation}
S_{\varepsilon}(U(0,x))=U^{per}(0,x;\varepsilon)+(1-e^{\frac{2\pi}{\omega} A})^{-1}\int_0^{\frac{2\pi}{\omega}}e^{(\frac{2\pi}{\omega}-s)A}\mathbf{f}(U(x,s;\varepsilon))\,ds,
\end{equation}
which is well defined by assumption (\ref{nonres}). This can be seen by inspecting equation (\ref{Uperfix}).\\
We can now use the estimate derived in 	(\ref{estimate}) to show that $S$ maps the ball $B_{\varepsilon}$ into itself. Indeed, we know from equation (\ref{Fourlin}), that $\|U_{per}(x,t;\varepsilon)\|_{\mathcal{H}}\leq \varepsilon(\|h(x,t)\|_{L^2(0,\pi)}+\|h_t(x,t)|\|_{L^2(0,\pi)})\lesssim\varepsilon$, where we also have used assumption (\ref{Asforce}). We find that 
\begin{equation}
\begin{split}
\|S(U(0,x))-U^{per}(0,x;\varepsilon)\|_{\mathcal{H}}&\lesssim \int_0^{\frac{2\pi}{\omega}}\|\mathbf{f}(U(x,s))\|_{\mathcal{H}}\,ds\\
& \lesssim \int_0^\frac{2\pi}{\omega} \|u\|_{L^2}^m\, ds\\
& \lesssim \varepsilon^m,
\end{split}
\end{equation}
where we have used assumption (\ref{AsA}) as well as the fact that $L^{2m}(0,\pi)\subset L^2(0,\pi)$ for $m>1$. For $\varepsilon$ sufficiently small, this proves the claim.\\
Now we show that the functional $S$ is a contraction on $B_\varepsilon$. To this end, note that
\begin{equation}
\left\|\frac{\partial S}{\partial U_0}\right\|\lesssim \int_0^{\frac{2\pi}{\omega}}\left\|\frac{\partial\mathbf{f}}{\partial U}\frac{\partial U}{\partial U_0}\right\|\,ds\lesssim \varepsilon^{m-1},
\end{equation}
where we have used the assumption on the derivative of $f$ in (\ref{AsA}), the fact that $\partial U/\partial U_0$ is close to $e^{A}$, and that $\varepsilon^{m-1}\lesssim \varepsilon$ for $\varepsilon$ small. Applying a standard fixed point argument then proves the lemma.  
\end{proof}

\section{Appendix II: Existence and Uniqueness of Spectral Submanifolds}
We now recall some general results on invariant submanifolds tangent to spectral subspaces from \cite{Cab2003}. Let $\mathcal{F}\in C^r(U,Y)$, with $r\in\mathbb{N}\cup\{\infty,a\}$ and let $0$ be a fixed point of $\mathcal{F}$. 
In the following, we denote the complex unit disk by
\begin{equation}
\mathbb{D}:=\{z\in\mathbb{C} : |z|<1\}.
\end{equation}

\begin{assumption}\label{AsB}
Let $\mathcal{A}$ be the derivative of the $C^r$-map $\mathcal{F}$ at zero, i.e., $\mathcal{A}=D\mathcal{F}(0)$. Assume further that\\
	\begin{enumerate}
		\item The operator $\mathcal{A}$ is invertible.
		\item The underlying Banach space $X$ admits a decomposition as a direct sum $X=X_1\oplus X_2$, where the space $X_1$ is invariant under $\mathcal{A}$, i.e.,
		$$\mathcal{A}X_1\subset X_1.$$
		We write $\pi_1:X\to X_1$ and $\pi_2:X\to X_2$ for the linear projections on the respective subspaces. For ease of notation, we set $\mathcal{A}_1:=\pi_1\mathcal{A}|_{X_1}$ and $\mathcal{A}_2:=\pi_2\mathcal{A}|_{X_2}$.
		\item The spectrum of $\mathcal{A}_1$ lies strictly inside the unit circle, that is to say 
		$$\sigma(\mathcal{A}_1)\subset \mathbb{D}.$$
		\item The spectrum of $\mathcal{A}_2$ does not contain zero, i.e., 
		$$0\notin\sigma(\mathcal{A}_2).$$
		\item For the smallest integer $L\geq 1$ with the property that 
		$$\sigma(\mathcal{A}_1)^{L+1}\sigma(\mathcal{A}_2^{-1})\subset\mathbb{D}$$
		we have
		$$\sigma(\mathcal{A}_1)^i\cap\sigma(\mathcal{A}_2)=\emptyset$$
		for every integer $i\in [2,L]$ (in case $L\geq 2$).
		\item The order of differentiability $r$ of $\mathcal{F}$ and the integer $L$ satisfy
		$$L+1\leq r.$$
		\end{enumerate}
\end{assumption}
\begin{remark}
	As a consequence of assumption $(2)$ in (\ref{AsB}), the operator $\mathcal{A}$ admits a representation  
	\begin{equation}\label{Amatrix}
	\mathcal{A}=\left(\begin{matrix} \mathcal{A}_1 & \mathcal{B}\\ 0 & \mathcal{A}_2\end{matrix}\right),
	\end{equation}
	with respect to the decomposition $X=X_1\oplus X_2$, where $\mathcal{B}=\pi_1\mathcal{A}|_{X_2}$. If $X_2$ is also an invariant subspace for $\mathcal{A}$, then $\mathcal{B}=0$. The main result in \cite{Cab2003} is the following
\end{remark}
\begin{theorem}\label{mainthm}
 Let $\mathcal{F}:U\to Y$ be a $C^r$-map that satisfies the assumptions (\ref{AsB}). Then the following holds true:
 \begin{enumerate}
\item There exists a $C^r$ manifold $\mathcal{M}_1$ that is invariant under $\mathcal{F}$ and is tangent to the subspace $X_1$ at $0$.
 \item The invariant manifold $\mathcal{M}_1$ is unique among all $C^{L+1}$ invariant manifolds of $\mathcal {F}$ that are tangent to the subspace $X_1$ at $0$. That is, every two $C^{L+1}$ invariant manifolds with this tangency property will coincide in a neighborhood of $0$.
 \item There exists a polynomial map $R:X_1\to X_1$ of degree not larger than $L$ and a $C^r$ map $K: U_1\to X$, defined on some open subset $U_1\subset X_1$ that contains $0$, satisfying
 $$R(0)=0,\quad DR(0)=\mathcal{A}_1,\quad K(0)=0,\quad \pi_1 DK(0)=I,\quad \pi_2 DK(0)=0$$ 
 such that $K$ serves as an embedding of $\mathcal{M}_1$ from $X_1$ to $X$ and $R$ represents the pull-back of the dynamics on $\mathcal{M}_1$ to $U_1$ under this embedding. Specifically, we have
 \begin{equation}\label{Conj}
\mathcal{F}\circ K=K\circ R.
 \end{equation}
  \end{enumerate}
\end{theorem}
The proof can be found in \cite{Cab2003}.
\begin{remark}
If additionally the non-resonance condition
	$$\sigma(\mathcal{A}_1)^i\cap\sigma(\mathcal{A}_2)=\emptyset$$
 holds for every integer $i\in [M,L]$, then we can choose $R$ in (\ref{Conj}) to be a polynomial of degree not larger than $M-1$. Furthermore, the $C^r$-manifold $\mathcal{M}_1$ is unique among all $C^{L+1}$ locally invariant manifolds tangent to the subspace $X_1$ at $0$ (see \cite{Cab2003} for details).
\end{remark}

\subsection{The Case of No External Forcing ($\varepsilon=0$)}
We will now apply Theorem \ref{mainthm} to our system (\ref{mainequ}) with $\varepsilon=0$. To this end, we choose as our underlying space $X=H^2(0,\pi)\times L^2(0,\pi)$ and as our map $\mathcal{F}=U(.,1)$, the time-one map of system $(\ref{mainequV})$.\\
 The mild solution $U=U(t,U_0)\in\mathcal{H}$, which is also a classical solution by Theorem \ref{ExUn}, is a fixed point to the function on the right hand side of equation (\ref{mild}). If we regard the right-hand side of (\ref{mild}), as the flow map $U_0\mapsto U_T(U_0):=U(T,U_0)$ for some fixed time $T>0$, we can take Fr\'{e}chet derivatives with respect to initial conditions $U_0=(u_0,v_0)\in\mathcal{D}(A)=H^4(0,\pi)\times H^4(0,\pi)$ on both sides to obtain
\begin{equation}
\begin{split}
\frac{\partial U_T}{\partial U_0}&=e^{TA}+\int_0^Te^{(t-s)A}\nabla_U F(u(s),x)\, ds\cdot\frac{\partial U_T}{\partial U_0}\\[0.3cm]
&=e^{TA}+\int_0^T e^{(t-s)A} \left(\begin{matrix} 0 & 0 \\ f'(u(s)) & 0 \end{matrix}\right)\, ds\cdot\frac{\partial U_T}{\partial U_0}.
\end{split}
\end{equation}
In particular, using the fact that $U_T(0,0)=0$ for all $T\geq 0$ by uniqueness of solutions and employing Assumption \ref{AsA}, we deduce that
\begin{equation}\label{derflow}
\frac{\partial U_T}{\partial U_0}(0,0)=e^{TA}
\end{equation}
for all $T\geq 0$.\\
This means that, from existence and uniqueness, we infer that $\mathcal{F}(0)=0$, while from equation (\ref{derflow}), we infer that $$\mathcal{A}=e^A.$$
Let us now verify Assumption $\ref{AsB}$ step by step.
To see that $\mathcal{A}$ is invertible, we may have a look at the explicit formula (\ref{explicit}) and compare the asymptotic growth of the coefficients. A look at (\ref{eigval}) confirms that $e^{\lambda_{n}^{\pm}}$ stays bounded and away from zero as $n\to\pm\infty$, so that we can deduce that the flow $e^{tA}$ maps the space $H^2(0,\pi)\times L^2(0,\pi)$ to itself. Since the map $(u_0,v_0)\mapsto e^{tA}(u_0,v_0)$ is bijective, as it can be seen by inspecting (\ref{explicit}), it follows from the Bounded Inverse Theorem on Banach spaces that the linearization $\mathcal{A}$ is invertible (cf. \cite{rudin2006functional}). In fact, we may write down the 
 of $\mathcal{A}$ in closed form as
\begin{equation}\label{explinv}
\begin{split}
\mathcal{A}^{-1} \left(\begin{array}{c}w\\z
\end{array}\right)=&\left(\sum_{n=1}^{\infty}e^{\lambda_n^++\lambda_n^-}\frac{(\lambda_n^-e^{\lambda_n^-}-\lambda_n^+e^{\lambda_n^{+}})\hat{w}_n-(e^{\lambda_n^-}-e^{\lambda_n^{+}})\hat{z}_n}{\lambda_n^--\lambda_n^+}\sin(nx),\right.\\[0.33cm]
&\left.
\sum_{n=1}^{\infty}e^{\lambda_n^++\lambda_n^-}\frac{-\lambda_n^+\lambda_n^-(e^{\lambda_n^+}-e^{\lambda_n^-})\hat{w}_n+(\lambda_n^-e^{\lambda_n^{+}}-\lambda_n^+e^{\lambda_n^-})\hat{z}_n}{\lambda_n^--\lambda_n^+}\sin(nx)\right).
\end{split}
\end{equation}\\
Since the space $L^2_0(0,\pi)$ admits a basis, namely $\{e^{i\tilde{n} x}\}_{n\in\mathbb{Z}}$, which is an eigenbasis for the right-hand side of equation (\ref{mainequ}), we may choose any subset of the eigenbasis of the matrix of operators $A$ as our space $X_1$. An easy computation shows that for a fixed frequency $e^{i\tilde{n} x}$, the eigenvectors for the matrix of operators $\mathcal{A}$ are given by $(1,\lambda_n^+)$ and $(1,\lambda_n^-)$, respectively. Therefore, for any subset $N\times M\subset \mathbb{N}^{+}\times\mathbb{N}^{+}$, we define the parametrization space as
\begin{equation}
X_1:=\spn\Big({\{(1,\lambda_n^+)\sin(nx)\}_{n\in N}}\cup{\{(1,\lambda_m^-)\sin(mx)\}_{m\in M}}\Big).
\end{equation}
The space $X_2$ in Assumption \ref{AsB} is then automatically given as $X_2=X_1^{\perp}$. We find that $X_1$ is indeed invariant under $\mathcal{A}$, taking a look at the explicit formula (\ref{explicit}) again.\\
Since the operator $\mathcal{A}$ is invertible and has pure point spectrum, we may deduce from the spectral mapping formula, cf. \cite{engel2006short}, that
\begin{equation} \sigma(\mathcal{A})=\left\{\exp\left(-\frac{\beta n^2+\delta}{2+2\mu n^2}\pm \sqrt{\left(\frac{\beta n^2+\delta}{2+2\mu n^2}\right)^2-\frac{\alpha n^4+\gamma}{1+\mu n^2}}\right)\right\}_{n\in\mathbb{N}^{+}}.
\end{equation}
\begin{remark}
Note that the above result is nontrivial in infinite dimensions. In general, the spectrum of the semi-flow generated by some operator $A:\mathcal{A}\subset\mathcal{H}\to\mathcal{H}$ is not equal to its exponential spectrum, i.e. $$\sigma(e^{tA})\neq e^{t\sigma(A)}$$ for some $t\geq 0$. There are examples of semi-groups that cannot be extended to operator groups, and therefore are not invertible.
\end{remark}
Since all constants in equation (\ref{mainequ}) are chosen positive, we immediately find that $\sigma(\mathcal{A})\subset\mathbb{D}$ and that $\sigma(\mathcal{A}_1)\subset\mathbb{D}$ for any choice of the subspace $X_1$. Thus, (2) of Assumption \ref{AsB} is automatically satisfied. By the same token, $0\notin\sigma(\mathcal{A}_2)$ and hence (3) of Assumption \ref{AsB} is always satisfied as well.\\
In order to give a criterion under which (5) of Assumption \ref{AsB} is satisfied, we first note that
$$(\sigma(\mathcal{A}_1))^{L+1}\sigma(\mathcal{A}_2^{-1})\subset\mathbb{D}$$
is satisfied if and only if 
$$ \sup \big\{-\operatorname{Re}\mu+(L+1)\operatorname{Re} \lambda\big\}<0,$$
where the supremum is taken over all $\lambda\in\{\lambda_n^+\}_{n\in N}\cup\{\lambda_m^-\}_{m\in M} $ and all $\mu\in \sigma(A)\setminus(\{\lambda_n^+\}_{n\in N}\cup\{\lambda_m^-\}_{m\in M})$.\\
Here, we have used the relation between the point spectrum of an invertible operator $\mathcal{A}$ and the point spectrum of its inverse $$\sigma_P(\mathcal{A}^{-1})\setminus\{0\}=\sigma_p(\mathcal{A})^{-1}.$$
A proof can be found in \cite{hislop2012introduction}.\\
Since $\lambda_{n}^{\pm}$ is negative for all $n\in\mathbb{Z}$, the above condition is equivalent to
\begin{equation}\label{ratio}
L\geq\frac{\inf\operatorname{Re} \mu}{\sup\operatorname{Re} \lambda}-1.
\end{equation}
This proves the existence of a SSM in the case of no external forcing.

\subsection{The Case of Weak External Forcing ($\varepsilon>0$)}
In the following, we will apply Theorem \ref{mainthm} to equation (\ref{mainequ}) with $\varepsilon>0$. Again, our underlying space will be $X=H^2_0(0,\pi)\times L^2_0(0,\pi)$. This time, however, we set $\mathcal{F}=P_\varepsilon$, with $P_\varepsilon$ being the Poincar\'{e} map defined in (\ref{Poin}). There exists then a fixed point of $P_\varepsilon$ (cf. Lemma \ref{fixed}), which we will denote by $U_\varepsilon^0$. 
We also let 
$$\mathcal{A}_\varepsilon=DP_\varepsilon(U_\varepsilon^0).$$
Since $e^{AT}$ is invertible and $\mathcal{A}_\varepsilon$ is close to $e^{AT}$ in norm for $\varepsilon$ small by Lemma \ref{derclose}, it follows that also $\mathcal{A}_\varepsilon$ is invertible.\\
Since we cannot write down the spectrum of $\mathcal{A}_\varepsilon$ explicitly, as we were able to do for $e^A$, we will aim to prove that $\sigma(\mathcal{A}_\varepsilon)$ is close to $\sigma(e^{AT})$. To this end, we will apply analytical spectral perturbation theory. Recall that, by the assumption in Theorem \ref{mainthmdifeps}, the nonlinearity in (\ref{mainequ}) is real analytic 
\begin{equation}\label{analytic}
f\in C^a(\mathcal{H},\mathcal{H}),
\end{equation}
which in particular implies that the for any $U\in \mathcal{H}$, the map $\varepsilon\mapsto \mathcal{A}_\varepsilon U$ is analytic around zero in the topology of $\mathcal{H}$. This follows by the Implicit Function Theorem for analytic functions on Banach spaces, cf. \cite{Lang1999}. Thus, $\mathcal{A}_\varepsilon$ defines a so-called analytic family and we can apply the following 

\begin{lemma}\textbf{(Analytical Spectral Perturbation)}\label{perspec}
	Let $T_\varepsilon: \mathcal{H}\to \mathcal{H}$ be an analytic family of bounded operators about $\varepsilon=0$. For any discrete eigenvalue $\lambda_0\in\sigma(T_0)$, there exist discrete eigenvalues $\lambda_1(\varepsilon),...,\lambda_r(\varepsilon)\in\sigma(T_\varepsilon)$, with $r=r(\varepsilon)$, such that $\lambda_j(0)=\lambda_0$, for $j=1,...,r$ and such that the total algebraic multiplicity
	 of the $\lambda_j$ is equal to the algebraic multiplicity of $\lambda_0$.
\end{lemma}
For a proof, which also includes the more general case of unbounded operators as well, we refer to \cite{hislop2012introduction}. Since the algebraic multiplicity is always greater or equal to the geometric multiplicity, a fortiori we know that sum of the dimensions of the eigenspaces associated with the split eigenvalues $\lambda_1(\varepsilon),...,\lambda_r(\varepsilon)$ cannot exceed the total multiplicity of $\lambda_0$.\\
From Lemma \ref{perspec}, we know that in a neighborhood of the eigenvalue $\lambda_0$ are eigenvalues $\lambda_1(\varepsilon),...,\lambda_r(\varepsilon)$ of the the perturbed operator $\mathcal{A}_{\varepsilon}$, which converge to the eigenvalue of the unperturbed operator as $\varepsilon\to 0$. The following Lemma guarantees that the spectrum is stable with respect to perturbations.
\begin{lemma}\label{perteig}
	Let $T_\varepsilon: \mathcal{H}\to \mathcal{H}$ be an analytic family of bounded operators about $\varepsilon=0$ and let $G$ be an open, bounded subset of the complex plane such that $\overline{G}\subset \rho(T_0)$. Then $G\subset \rho(T_\varepsilon)$ for $\varepsilon$ sufficiently small.
\end{lemma}
A proof can be found in \cite{hislop2012introduction}. Choosing now as our $G$ the whole of $\mathbb{C}$ with small balls around the discrete eigenvalues of $e^A$ excluded, we can deduce that
\begin{equation}\label{perD}
\sigma(\mathcal{A}_\varepsilon)\subset\mathbb{D},
\end{equation}
as desired for the application of Theorem \ref{mainthm}.\\
To choose a parametrization space $X_1^\varepsilon$, we introduce the so-called \textit{Riesz projection}. Let $\Sigma_\varepsilon$ be a collection of isolated eigenvalues of the operator $\mathcal{A}_\varepsilon$ and let $\Gamma:[0,1]\to\mathbb{C}$ be a simply-closed curve with winding number one that encircles $\Sigma_\varepsilon$ and does not intersect with the remaining spectrum of $\mathcal{A}_\varepsilon$. Define the operator-valued function $\varepsilon\mapsto \mathbb{P}_\varepsilon$,
\begin{equation}
\mathbb{P}_\varepsilon:=\oint_\Gamma(z-\mathcal{A}_\varepsilon)^{-1}\,dz.
\end{equation} 
The operator $\mathbb{P}_\varepsilon$ defines a projection, cf. \cite{hislop2012introduction}. Moreover, the underlying Hilbert space admits a splitting as
\begin{equation}
	\mathcal{H}=\ker(\mathbb{P}_\varepsilon)\oplus\Ran(\mathbb{P}_\varepsilon),
\end{equation}
both $\ker(\mathbb{P}_\varepsilon)$ and $\Ran(\mathbb{P}_\varepsilon)$ are invariant spectral subspaces for the operator $\mathcal{A}_\varepsilon$, i.e.
\begin{equation}
\begin{split}
&\sigma(\left.\mathcal{A}_\varepsilon\right|_{\Ran(\mathbb{P}_\varepsilon)})=\Sigma_\varepsilon,\\
&\sigma(\left.\mathcal{A}_\varepsilon\right|_{\ker(\mathbb{P}_\varepsilon)})=\sigma(\mathcal{A}_\varepsilon)\setminus\Sigma_\varepsilon.
\end{split}
\end{equation}
A proof for the above result, the \textit{Riesz decomposition Theorem}, can be found in \cite{gohberg2013classes}. 
Now, for any collection os split-eigenvalues $\Sigma_\varepsilon$ of the operator $\mathcal{A}_\varepsilon$, we set
\begin{equation}\label{Riesz}
\begin{split}
&X^{\varepsilon}_1:=\Ran(\mathbb{P}_\varepsilon),\\
&X^{\varepsilon}_2:=\ker(\mathbb{P}_\varepsilon).
\end{split}
\end{equation}
Equally, the parametrization space can be written as
\begin{equation}
X_1^{\varepsilon}=\bigoplus_{n\in N} \bigoplus_{k(n)}\Eig(\lambda_{n}^{k(n)}(\varepsilon)),
\end{equation}
where $\{\lambda_n\}_{n\in N}\subset \sigma(e^A)$ is some subset of the spectrum of the unperturbed linear flow map, while the index $n\mapsto k(n)$ describes some choice of the split eigenvalues for the perturbed operator.

By the above considerations, the space $X_1^\varepsilon$ is invariant and clearly $\sigma(\mathcal{A}_1^\varepsilon)\subset\mathbb{D}$ by relation (\ref{perD}). Also, since $0\notin\sigma(\mathcal{A}_\varepsilon)$ for small enough $\varepsilon$ and since the operator $\mathcal{A}_\varepsilon$ decomposes according to (\ref{Amatrix}), we conclude that $0\notin\sigma(\mathcal{A}_{2}^{\varepsilon})$.\\
The eigenvalues of a family of self-adjoint operators $\{T_\varepsilon\}_{\varepsilon}$, depending analytically on $\varepsilon$, also depend analytically upon $\varepsilon$, cf. \cite{hislop2012introduction}. However, in our case, the analytic family $\{\mathcal{A}_\varepsilon\}$ is not self-adjoint and in general each eigenvalue $\lambda_j(\varepsilon)$ is only analytic in $\varepsilon^{1/p}$ for some integer $p$, i.e., $\lambda_j$ possesses a Puiseux expansion in $\varepsilon$. However, all isolated eigenvalues of $\mathcal{A}_0$ being non-degenerate, the following Theorem applies in our case.
\begin{theorem}\label{ana}
Let $\{T_\varepsilon\}_\varepsilon$ be an analytic family of type A about $\varepsilon=0$. Let $\lambda$ be a discrete, non-degenerate eigenvalue of $T_0$. Then there exists an analytic family, $\lambda(\varepsilon)$, of discrete, non-degenerate eigenvalues of $T_\varepsilon$, such that $\lambda(0)=\lambda$, for $|\varepsilon|$ sufficiently small.\\
Moreover, the associated Riesz projections, $\varepsilon\mapsto\mathbb{P}_\varepsilon$ depend analytically on $\varepsilon$.
\end{theorem}
A proof can be found in \cite{hislop2012introduction}.
\begin{remark}
If the eigenvalue $\lambda$ is degenerate, then Theorem \ref{ana} in general fails, as already the finite dimensional example $T_\varepsilon:\mathbb{R}^2\to\mathbb{R}^2,$ $$T_\varepsilon=\left(\begin{matrix}
2 & 1\\ \varepsilon & 2
\end{matrix}\right),$$ shows. Indeed, the spectrum of $T_0$ consists of the eigenvalue $\lambda=2$ with algebraic multiplicity two, while $\sigma(T_\varepsilon)=\{2\pm\sqrt{\varepsilon}\}$. Using a standard residue calculus argument, one readily finds that the Riesz projection for the eigenvalue $\lambda(\varepsilon)=2+\sqrt{\varepsilon}$ is given by $$\mathbb{P}_\varepsilon=\oint_{\Gamma}\frac{1}{(z-2)^2-\varepsilon}\left(\begin{matrix}
z-2 & 1 \\ \varepsilon & z-2
\end{matrix}\right)\, dz=\frac{1}{2}\left(\begin{matrix} 1 & \frac{1}{\sqrt{\varepsilon}}\\ \sqrt{\varepsilon} & 1\end{matrix}\right),$$ which is not analytic in $\varepsilon$ around zero.
\end{remark}

\begin{proposition}\label{perspace}
The spectral subspace $\mathcal{E}_\varepsilon$, associated to the set of perturbed eigenvalues $\Sigma_\varepsilon\subset\sigma(\mathcal{A}_\varepsilon)$, is $\varepsilon$-close to the spectral subspace $\mathcal{E}$, associated to the collection of eigenvalues $\Sigma\subset\sigma(\mathcal{A}_0)$ of the unperturbed operator $\mathcal{A}$, i.e.,
\begin{equation}
\mathcal{E}_\varepsilon=\mathcal{E}+\mathcal{O}(\varepsilon).
\end{equation}
\end{proposition}
\begin{proof}
Since, by Theorem \ref{ana}, the perturbed Riesz projection is analytic in $\varepsilon$, i.e., $\mathbb{P}_\varepsilon=\mathbb{P}_0+\mathcal{O}(\varepsilon)$, it follows that
\begin{equation}\label{SSMeps}
\mathcal{E}_\varepsilon=\Ran(\mathbb{P}_\varepsilon)=\Ran(\mathbb{P}_0)+\mathcal{O}(\varepsilon)=\mathcal{E}+\mathcal{O}(\varepsilon)
\end{equation}
by the Riesz projection Theorem and equation (\ref{Riesz}). This proves the claim.
\end{proof}
Since all eigenvalues of $\mathcal{A}_0$ are simple by equation \ref{specA}, Proposition \ref{SSMeps} guarantees that the perturbed spectral subspaces $\mathcal{E}_\varepsilon$ are $\varepsilon$-close to the unperturbed eigenspaces of $\mathcal{A}_0$.\\
For condition (5) in Theorem \ref{mainthm}, we have a look at the relation obtained in (\ref{ratio}) and note that it will hold true if we perturb $\lambda$ and $\mu$ only slightly, that is to say for small enough $\varepsilon$. Since we assumed (\ref{analytic}), condition (6) in Theorem \ref{mainthm} is automatically satisfied for any $L\in\mathbb{N}$.\\
We deduce from Theorem \ref{mainthm} that there exists a unique, analytic invariant manifold for the Poincar\'e map tangent to the spectral subspace $X_1^\varepsilon$. Since the choice of the base point $\omega_0$ in the definition of the Poincar\'e map (\ref{Poin}) was arbitrary, we obtain an invariant manifold for any such $\omega_0\in S_\omega$. By analytical dependence upon initial conditions for the flow map, \textit{a forteriori}, the Poincar\'{e} map dependence analytically upon $\omega_0$. Because of the uniqueness of the invariant manifold, obtained for any $\omega_0\in S_\omega$, we deduce that there exists a unique, invariant manifold for the flow of equation (\ref{suspend}) by continuing the spectral submanifold for any $\theta\in S_\omega$. 

\end{document}